\documentclass[12pt]{amsart}

\def\eps{{\varepsilon}}

\def\Prob{{\mathbb{P}}}
\def\prob{{\mathbb{P}}}

\def\reals{\mathbb{R}}

\def\integers{\mathbb{Z}}

\def\bA{\mathbf{A}}

\def\bB{\mathbf{B}}

\def\bG{\mathbf{G}}

\def\bS{\mathbf{S}}

\def\bT{\mathbf{T}}

\def\bZ{\mathbf{Z}}

\def\bt{\mathbf{t}}

\def\bx{\mathbf{x}}

\def\brL{{\bar L}}

\def\brc{{\bar c}}

\def\cA{\mathcal{A}}

\def\cB{\mathcal{B}}

\def\cC{\mathcal{C}}

\def\cD{\mathcal{D}}

\def\cI{\mathcal{I}}

\def\cJ{\mathcal{J}}

\def\cF{\mathcal{F}}

\def\cG{\mathcal{G}}

\def\cE{\mathcal{E}}

\def\cL{\mathcal{L}}

\def\cM{\mathcal{M}}

\def\cN{\mathcal{N}}

\def\cP{\mathcal{P}}

\def\cS{\mathcal{S}}

\def\cT{\mathcal{T}}

\def\cW{\mathcal{W}}

\def\cX{\mathcal{X}}

\def\cY{\mathcal{Y}}
\def\cZ{\mathcal{Z}}

\def\fM{\mathfrak{M}}

\def\fN{\mathfrak{N}}

\def\fg{\mathfrak{g}}

\def\fl{\mathfrak{l}}

\def\tF{{\tilde F}}

\def\tL{{\tilde L}}

\def\tk{{\tilde k}}

\makeatother

\def\beq{\begin{equation}}
\def\eeq{\end{equation}}

\usepackage{geometry}                		
\usepackage{graphicx}				
\usepackage{amssymb}
\usepackage{amsmath}
\usepackage{mathtools}
\usepackage{ragged2e}
\usepackage{attachfile}
\usepackage{graphicx}
\usepackage{subcaption}
\usepackage{float}
\usepackage{tikz}
\usetikzlibrary{shapes.geometric}
\usetikzlibrary{calc}
\usetikzlibrary{spy}
\usepackage[utf8]{inputenc}
\usepackage[english]{babel}
\usepackage{amsthm}
\usepackage{mathtools,amsthm}
\usepackage{mathtools}

\usepackage{amsmath,amssymb}
\usepackage{amssymb}

\newtheoremstyle{case}{}{}{}{}{}{:}{ }{}
\theoremstyle{case}

\theoremstyle{definition}
\newtheorem{definition}{Definition}[section]
\newtheorem{theorem}[definition]{Theorem}
\newtheorem{lemma}[definition]{Lemma}

\newtheorem{proposition}[definition]{Proposition}

\newtheorem{remark}[definition]{Remark}
\newtheorem{example}[definition]{Example}
\numberwithin{equation}{section}

\pgfmathsetmacro{\xcoord}{cos(60)}
\pgfmathsetmacro{\ycoord}{sin(60)}

\newcommand{\boundellipse}[3]
{(#1) ellipse (#2 and #3)
}

\usetikzlibrary{calc}
\usetikzlibrary{spy}
\usetikzlibrary{calc,trees,positioning,arrows,chains,shapes.geometric,%
    decorations.pathreplacing,decorations.pathmorphing,shapes,%
    matrix,shapes.symbols}

\tikzset{
>=stealth',
  punktchain/.style={
    rectangle, 
    rounded corners, 
    draw=black, very thick,
    text width=10em, 
    minimum height=3em, 
    text centered, 
    on chain},
  line/.style={draw, thick, <-},
  element/.style={
    tape,
    top color=white,
    bottom color=blue!50!black!60!,
    minimum width=8em,
    draw=blue!40!black!90, very thick,
    text width=10em, 
    minimum height=3.5em, 
    text centered, 
    on chain},
  every join/.style={->, thick,shorten >=1pt},
  decoration={brace},
  tuborg/.style={decorate},
  tubnode/.style={midway, right=2pt},
}

\author{P\'eter N\'andori}
\address{Department of Mathematical Sciences, Yeshiva University, New York, NY, USA}
\email{peter.nandori@yu.edu }
\author{Trevor Teolis}
\address{Department of Mathematics, University of Illinois
at Chicago, Chicago, IL, USA}
\email{tteoli2@uic.edu}
\title{Local equilibrium in planar 
non interacting particle systems}
\date{}							

\begin{document}
\maketitle

\begin{abstract}
Particles are injected to a large planar rectangle through the boundary. Assuming
that the particles move independently from one another
and the boundary is also absorbing, 
we identify a set of abstract conditions which imply the local equilibrium of 
the particle density in 
diffusive scaling limit. We verify that our abstract conditions hold in two examples:
iid random walks and the periodic Lorentz process. 
\end{abstract}


\section{Introduction}

A major open problem in mathematical statistical mechanics is to rigorously 
derive macroscopic
laws of physics, such as Fourier's law of heat conduction, from underlying microscopic
principles \cite{BLRB00}.
A realistic microscopic model should consist of a macrosopic domain inside which
the microscopic particles are subject to some bulk dynamics and interact with a heat bath
on the boundary. If the temperature of the heat bath varies along the boundary,
then one would like to study the emergence of local equilibrium (i.e. the existence of a well 
defined temperature at microscopic or mesoscopic locations inside the domain).

There are two separate classes of models for the particle dynamics. The first class
is stochastic, namely Markov processes. Because of the Markov property, the future
of the system can equally be described no matter what happened in the past and 
so Markov processes provide an excellent opportunity to derive beautiful mathematical
results. Indeed, the results oftentimes go much beyond the derivation of the
heat equations (such as second order fluctuations or other PDEs).
We don't attempt to review the literature of such Markov models here so we refer the reader 
to the classical surveys \cite{S91, KL}.

The second class of models are realistic (Hamiltonian) deterministic dynamical 
systems. Proving that the bulk 
dynamics obey the heat equation becomes considerably harder for such deterministic systems.  
However, a notable realistic Hamiltonian system for which rigorous results are available is the Sinai 
billiard \cite{S70}. 
In Sinai billiards, point particles fly freely among fixed convex bodies and elastically 
collide on their boundaries. In this case, a rigorous study of a variant
of the problem prescribed
in the first paragraph is possible when 'temperature' is replaced by 'particle density' 
and
'the heat bath' is replaced by 'varying chemical potential'. Indeed, the point 
particles do not interact with one another and so there is no exchange of energies. Furthermore, the trajectory of each particle satisfies the central limit theorem
 \cite{BS81, 
BSCh90, BSC91} leading to the heat equation in the bulk. 
However, even in this case of a 'non-interacting particle system', 
a better understanding of the boundary phenomena is desirable. 

We now describe the problem to be studied here. Let $D \subset \reals^2$ be a bounded domain with piece-wise smooth
boundary and
let particles be injected to the large domain $L D$ 
for $L \gg 1$
through
its boundary.
The particles will then perform some independent motion $\cZ$ on a lattice inside $L D$.
The boundary is also absorbing so most particles
are killed (i.e. absorbed) shortly after injection. However, some will survive for 
a long time and find their way deep into the interior of $L D$. 
The problem now is to show that the
limiting density profile of particles is governed
by the heat equation when time is rescaled by $L^2$
and by the Laplace
equation when time is infinite, where, in both cases, the boundary conditions are given by the injection rate.
We will refer to the first case as hydrodynamic and the second
one as hydrostatic limit.
This terminology is somewhat unusual since there is no energy exchange here, but we find it natural since we are studying the scaling 
limits of particle systems.
Specifically, we look at the problem
of proving local equilibrium of the particle density profile in systems forced
out of equilibrium when the particle injection rate varies along
the boundary of the domain. 

In this paper, we identify an abstract framework for which we can solve the problem presented in the previous paragraph, that is, we can prove the local equilibrium in both the hydrodynamic 
and the hydrostatic limit in the case $D$ is a rectangle.  This framework is general enough to include two basic examples: (1) when $\cZ$ is an
iid random walk, and (2) when $\cZ$ is given by the 
spatially periodic extension of the 
Sinai billiard (called periodic Lorentz process).
The abstract
framework is given by some hypotheses (H1)-(H3)
(see Section \ref{sec:def1}). The main hypothesis
is (H2), which is a conditional local invariance principle conditioned on 
survival of the particle. In case of the periodic Lorentz gas, our
result provides a natural extension of \cite{DN16} from one dimensional 
domains (i.e. line segments) to two dimensional rectangles. 
See Figure \ref{fig2} for the case of periodic Lorentz gas: particles, 
indicated by blue dots, are 
injected from the left ("West") side of a large rectangle while the entire
boundary of the rectangle is absorbing.

\begin{figure}
\begin{center}
\begin{tikzpicture}[scale=1.2]
\centering


\foreach \j in {0,1,2,3,4,5}
{
\foreach \i in { -3, -2, -1,0,1, 2,3,4,5,6}
{
\draw[fill=gray!50] \boundellipse{\i+0.5,\j + 0.5}{0.3}{0.4};
\path[fill=gray!50] (\i+ 0.7,\j + 0) to [out=90, in=180] (\i+1,\j + 0.3)  --  (\i+1,\j +  0)--  (\i+ 0.7,\j + 0);
\path[fill=gray!50] (\i+ 0.7,\j + 1) to [out=270, in=180] (\i+1,\j + 0.7) --  (\i+1, \j + 1) -- (\i+ 0.7,\j + 1);
\path[fill=gray!50] (\i+ 0.3,\j + 0) to [out=90, in=0] (\i,\j + 0.3)  --  (\i, \j + 0)--  (\i+ 0.3,\j + 0);
\path[fill=gray!50] (\i+ 0.3,\j + 1) to [out=270, in=0] (\i,\j + 0.7) --  (\i,\j +  1) -- (\i+ 0.3,\j + 1);

\draw (\i,\j) -- (\i+1,\j) -- (\i +1 ,\j + 1) -- (\i,\j + 1) --(\i,\j);

\draw (\i,\j + 0.7) arc (270:360:0.3);
\draw (\i+0.7,\j + 1) arc (180:270:0.3);
\draw (\i+1,\j + 0.3) arc (90:180:0.3);
\draw (\i+0.3,\j + 0) arc (0:90:0.3);
}
}


\foreach \j in {0,1,2,3,4,5}
{
\foreach \i in { -3, -2, -1}
{
\draw (\i + 0.1 + rand*0.1,\j  + 0.5 + rand*0.2) node[draw,circle,fill=blue,blue,inner sep=0,minimum size=0.6mm] (W) {};
}
}

\foreach \j in {1,2,3,4}
{
\foreach \i in { -3, -2, -1,0,1}
{
\draw (\i + 0.9 + rand*0.1,\j  + 0.5 + rand*0.2) node[draw,circle,fill=blue,blue,inner sep=0,minimum size=0.6mm] (W) {};
}
}

\foreach \j in {2,3}
{
\foreach \i in { -3, -2, -1,0,1,2,3}
{
\draw (\i + 0.9 + rand*0.1,\j  + 0.5 + rand*0.2) node[draw,circle,fill=blue,blue,inner sep=0,minimum size=0.6mm] (W) {};
}
}

\foreach \j in {2,4}
{
\foreach \i in { -1,1}
{
\draw (\i + 0.5 + rand*0.2,\j  + 1 + rand*0.07) node[draw,circle,fill=blue,blue,inner sep=0,minimum size=0.6mm] (W) {};
}
}

\foreach \j in {2,3,4,5}
{
\foreach \i in { -3, -2, -1,0}
{
\draw (\i + 0.1 + rand*0.1,\j  + 0.5 + rand*0.2) node[draw,circle,fill=blue,blue,inner sep=0,minimum size=0.6mm] (W) {};
}
}

\draw (4+ 0.1 + rand*0.1,4  + 0.5 + rand*0.2) node[draw,circle,fill=blue,blue,inner sep=0,minimum size=0.6mm] (W) {};

\draw (5+ 0.9 + rand*0.1, 1  + 0.5 + rand*0.2) node[draw,circle,fill=blue,blue,inner sep=0,minimum size=0.6mm] (W) {};

\end{tikzpicture}
\caption{
Particle configuration in a large rectangle
(point particles are enlarged for better visibility)
} \label{fig2}
\end{center}
\end{figure}
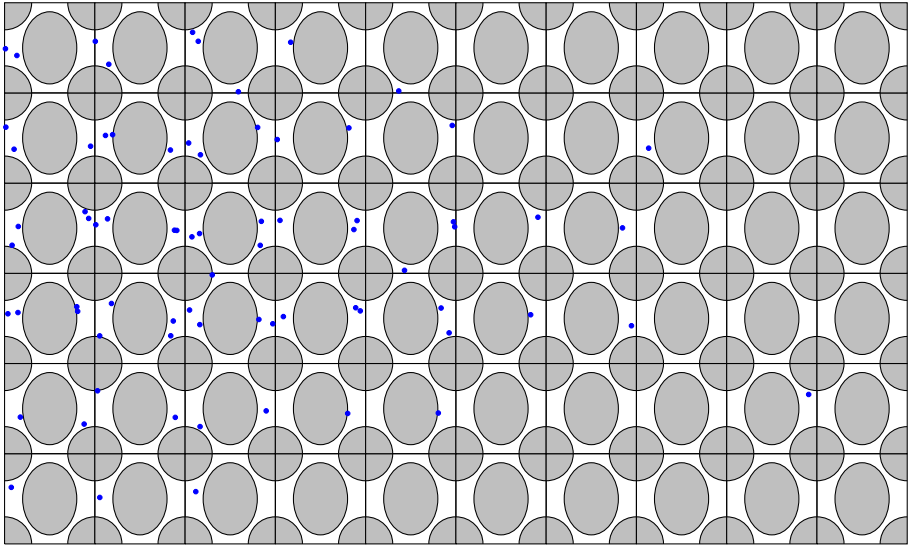

Complementary to our approach, there is a classical proof based
on the idea of duality. If the particle motion in the bulk has a nice
dual process and the injection on the boundary is chosen very carefully,
then the problem can be reformulated in terms of the hitting of $\partial D$ by 
the dual process starting from the bulk.
The approach by duality thus gives similar results with two major differences:
it is more
general in the sense that $D$ can be any domain with piece-wise smooth 
boundary, but it is more restrictive in the sense that requires both the 
existence of a nice
dual process and a very specific way of injection on the boundary.

In case of our two basic examples, the dual process 
(in fact a certain reverted process) is essentially the
same as the original process.
We present the proof by duality in 
Section \ref{sec:dual}.
In case of Markov processes, the proof by duality is very well
known even for some interacting particle systems, see e.g.
\cite{KMP} (we do not 
review the literature and do not
study any of these stochastic interacting particle systems here).
The proof by duality is not surprising for the Lorentz gas either, but it
was not observed in \cite{DN16} and so our Proposition
\ref{prop:dualbil} (with trivial changes to include
a $1$ dimensional macroscopic domain)
gives a simple new proof of the main results of \cite{DN16} in case of a
very special injection mechanism, which is essentially given by the Lebesque measure.
The utility of this special injection mechanism is limited, however, since no reasonable heat bath is 
likely to preserve the invariant measure of the bulk dynamics (see e.g. \cite{B..12}).

Because the proof by duality requires a very rigid structure of both 
the bulk dynamics and the injection, it is essential to develop other tools which do not require
such a rigid structure. Such tools are exemplified by our main results in Sections \ref{sec2} and \ref{sec3}.
Indeed, our injection procedure \eqref{eq:Poi} is quite general:
besides the dependence on the macroscopic position we allow 
the injection rate to depend on
the microscopic geometry through some function $\bA$
and on time through another function
$\bB$.
In case of deterministic systems, the only source of randomness is the
choice of the initial condition
according to an initial probability measure. Once the initial condition is fixed, $\cZ$ is deterministic.
We allow a lot of initial measures. For example any 
"standard pair" \cite{CD09} in case of Sinai billiards. In our context a useful way of thinking about standard pairs is
that they are conditional measures corresponding 
to a given past symbolic trajectory of the particle. Except for the special
choice of $\bA$ and $\bB$ as in Proposition \ref{prop:RW}, 
we believe that our results are new even in case of random walks.
Finally, we believe that some ideas presented here could be of use for studying deterministic interacting particle
systems as well.

The rest of this paper is organized as follows.
In Section \ref{sec2}, we provide the basic definitions 
and the main result Theorem \ref{thm1} in our abstract
framework. In Section \ref{sec3}, we present our two basic examples
-- namely, the random walk and the Lorentz gas.
In Section \ref{sec:dual}, we discuss the approach by duality. 
In Section \ref{sec5}, we prove Theorem \ref{thm1}.
The most technical part of this work is 
the verification of the conditional local invariance principle
(H2) for the Lorentz gas, which is presented in Section \ref{sec6}. This section
is heavily built upon the 
standard pair technique of Chernov and Dolgopyat 
\cite{CD09} and tools from \cite{DN16}, such as the mixing local limit theorem.
The necessary background is summarized in 
Section \ref{sec:billb}.

\section{Abstract setup}
\label{sec2}

\subsection{Non-interacting particle systems}
\label{sec:def1}

Let $\cL \subset \reals^2$ be a lattice of dimension $2$. 
We consider the graph $\bG$ with vertices $\cL$ and edges joining $l$
with $l+w_j$ for all $l \in \cL$ and $j=1,...,J$ for a fixed set 
$\{ w_1,...,w_J \} \subset \cL$.
For $z \in \reals^2$, let $\langle z \rangle$ be the closest $\fl \in \cL$ to $z$
with the property that $\fl_1 \geq z_1$
(if there is more than one such lattice points,
then choose the smallest in lexicographic order).

Let 
$(\bS, \mathbb P)$ be a probability space and 
$\cZ_t$  ($t \geq 0$) be an $\cL$ valued stochastic
process. That is, $\cZ_t: \bS  \rightarrow \cL$ for every $t \geq 0$. We assume that $\cZ$ is 
continuous from the right and has left limits. In other words, $\cZ$ is
a c\`adl\`ag function (i.e. for almost every
$\bf s \in \bS$ fixed, $\cZ$ jumps at random
times $\bf t$ from a lattice point $\cZ_{\bf t-}$ to another lattice point $\cZ_{\bf t}$).
We do not assume that $\cZ$ is Markovian.

Now let $ D= [0,A] \times [0,1]$ for a fixed positive real $A$.
Fix a non-negative continuous
functions $f: [0,1] \to \reals$
and write
$F: \partial D \to \reals$, 
$$
F(z) = 
\begin{cases}
f(y) \text{ if } z = (0,y)\\
0 \text{ otherwise. }
\end{cases}
$$ 
We will consider the following Dirichlet problems
\begin{equation}
\label{eq:laplace}
\Delta u = 0, \quad u|_{\partial D} = \varsigma F,
\end{equation}
\begin{equation}
\label{eq:heat}
v_t = \frac12 \left[ v_{xx} + v_{yy} \right], \quad
v(t,x,y)|_{(x,y) \in \partial D} = \varsigma F,
v(0,x,y) = 0.
\end{equation}

We are actually interested in the Dirichlet problem 
where $F$ is permitted to be nonzero for all boundary points (as in \eqref{eq:laplace2}),
but this case follows from linearity.
By classical theory, there is a unique solution to both the Laplace
equation
\eqref{eq:laplace} and the heat equation \eqref{eq:heat}
and furthermore $\lim_{t \to \infty}v(t,x,y) = u(x,y)$.
Of course this is true for much more general domains $D$,
e.g. when $\partial D$ is piecewise-smooth with no cusps.


For $L \gg 1$, let $ D_L = (LD) \cap \cL$, 
$$\partial D_L = \{ \fl \in D_L: \fl \text{ is connected to a point outside of $D_L$} \},$$
and
$$\partial_W D_L = \{ \fl \in D_L: \fl \text{ is connected to a point $\fl'$ with $\fl'_1 <0$} \}.$$
Here $\partial_W$ stands for West boundary as points in $\partial_W D_L$ are close
to the "West" side of the rectangle $D_L$.
Given $\fl \in \partial D_L$, let 
$$\cJ(\fl)
= \{ j=1,...,J: \fl + w_j \notin D_L \}
$$
We consider the following process for $L \gg 1$.
First, for some $t \in \mathbb R_+ \cup \{ \infty\}$, 
let $\Theta_t$ be a Poisson point 
process on $(-t, 0] \times \partial_W D_L$
with intensity measure 
\begin{equation}
\label{eq:Poi}
\bA(\cJ(\fl)) \bB(s) f(\fl_2/L) d \mbox{Leb} (s) 
d \mbox{counting} (\fl),
\end{equation}
where $\fl \in \partial D_L$
and $\bA: 2^{\{1,...,J \}}  \to \reals_+$ and $\bB: \reals \to \reals_+$ are fixed functions.
We assume that $\bB$ is continuous, periodic 
with period $1$, and $\int_0^1 \bB = 1$. 
One example is $\bA(\cJ)  = |\cJ|$ and $\bB = 1$. However, we want
to allow more general functions to accommodate for more general behavior
of the heat bath.

For each point $(T,\fl) \in \Theta$, we start an iid copy
of $\mathcal Z$ at time $T$ from position $\fl$ and we kill
it at 
\begin{equation}
\label{eq:tau*}
\tau^* = \inf \{ t > T : \cZ_t \notin D_L \},
\end{equation}
the first exit from $D_L$. 
In the case $\cZ$ is not Markovian, the initial condition $\cZ_T = \fl$
may not define the distribution of $\cZ_{T+t}$ for $t > 0$ uniquely. 
In this case, we allow 
multiple choices of this distribution but we require that 
$\cZ_{T+t} -\fl $ only depends on $\fl$ through 
$\cJ(\fl)$. That is, if $\fl, \fl' \in \partial D_L$,  with 
$\cJ(\fl) = \cJ(\fl')$ and $(T, \fl), (T', \fl') \in \Theta$, then we require
that for all $t \geq 0$, and for all $\tilde{\fl} \in \cL$,
$$\Prob(\cZ_{T+t} = \fl + \tilde{\fl} | \cZ_T = \fl)
= \Prob(\cZ_{T'+t} = \fl' + \tilde{\fl} | \cZ_{T'} = \fl').
$$
This procedure is to be interpreted
as injecting a particle to the domain $D_L$ at time $T$ 
through an edge $(\fl_-, \fl)$ of the graph $\bG$, where $\fl_- \notin D_L$,
$\fl \in D_L$
and letting particles evolve independently from one another 
until coming back to the absorbing boundary. The specific mechanism of injection through $(\fl_-, \fl)$ only depends on $j = 1,...,J$, where
$\fl - \fl_- = w_j$.
Let $\Lambda_t (\fl)$ be the number of particles at site $\fl$ at time $T=0$. We start with the 
following abstract result.

\begin{theorem}
\label{thm1}
Assume that (H1) - (H3) are satisfied. Then 
for any $z$ in the interior of $D$
\begin{equation}
\label{eq:thm11}
\lim_{L \to \infty}\mathbb E (\Lambda_{\infty} (\langle z L  \rangle)) = u(z)
\end{equation}
and
\begin{equation}
\label{eq:thm12}
\lim_{L \to \infty}\mathbb 
E (\Lambda_{tL^2} (\langle z L  \rangle)) = v(t,z)
\end{equation}
where $u$ and $v$ are defined by 
\eqref{eq:laplace} and \eqref{eq:heat}
with some $\varsigma$. 
\end{theorem}

To define our hypotheses (H1) - (H3), we need some definitions.

\noindent Let $W_t$ be a standard Brownian motion. 
Let 
$$\phi(\eta, \gamma, \xi)
= \lim_{dt \to 0} \frac{1}{dt} 
\prob (W_1 \in [\gamma, \gamma + dt],
\min_{t \in [0,1]} W_t >0,\max_{t \in [0,1]} W_t < \xi | W_0 = \eta).
$$ 
It is known (see e.g. \cite{F71})
that for any $0<\gamma, \eta < \xi$, the following formula holds
\begin{equation}
\label{eq:phi}
\phi(\eta, \gamma, \xi) = \sum_{n=-\infty}^{\infty} \frac{1}{\sqrt{2\pi}} \Bigg( \exp \Big(-\frac{(\gamma - \eta - 2n\xi)^2}{2} \Big) - \exp \Big(-\frac{(\gamma + \eta + 2n\xi)^2}{2}\Big) \Bigg).
\end{equation}

Recall that the Brownian meander is a stochastic
process on $[0,1]$ obtained by conditioning a standard Brownian 
motion to stay positive on $[0,1]$ (which has zero probability 
but the definition still makes sense by conditioning on 
staying above $-\eps$, letting $\eps \to 0$ and taking weak limit, see e.g. \cite{DIM77}).
Let $\mathfrak X(t)$ be a Brownian meander and $\mathfrak M(t) = \max_{0\leq s \leq t} \mathfrak X(s)$ its maximum. Then it is proven in \cite[Theorem 5]{Ch76} that the function
$$
\psi(\alpha, \beta) = \lim_{dt \to 0} \frac{1}{dt} \Prob (\mathfrak X(1) \in [\alpha, \alpha+ dt], \mathfrak
M(1) < \beta)
$$
for any $0 < \alpha < \beta$ satisfies 
\begin{equation}
\label{eq:psi}
\psi(\alpha, \beta) = \sum_{k=-\infty}^{\infty} (2k\beta + \alpha)\exp \Big(-\frac{(2k\beta + \alpha)^2}{2} \Big).
\end{equation}
Note that the formulas \eqref{eq:phi} and \eqref{eq:psi}
are closely related as the Brownian meander is closely related to
the Brownian motion. Indeed, by the definition of Brownian meander,
$\psi(\alpha, \beta) = 
\lim_{\eta \to 0} 
\phi(\eta, \alpha,\beta) / \int_0^\beta \phi(\eta, \alpha',\beta) d\alpha'$. We refer to \cite{Ch76} for more details.

Let us write $\cZ_t = (\cX_t, \cY_t)$. Denote
$$
\tau_{x}^{\cX} = 
\begin{cases}
\min\{ 
t>0: \cX_t > x \} & \text{ if } x >0\\
\min\{ 
t>0: \cX_t < x \} & \text{ if } x \leq 0.
\end{cases}
$$
We define $\tau_{y}^{\cY}$ analogously.

Now we make the following assumptions:
\begin{enumerate}
\item[(H1)] {\bf Vertical rational dependence}
{\it There is some $\fl \in \cL$, $\fl \neq 0$ so that $\fl_1 = 0$.}
\end{enumerate}

Let $(0, 0) = \fl^{(0)}, \fl^{(1)},\fl^{(2)}...$ be the enumeration of points $\fl \in \partial D_L$
which are connected to lattice points with negative first coordinate
in increasing order of second coordinate (that is 
$\fl_2^{(j)} \leq \fl_{2}^{(j+1)}$). If there are points 
$\fl^{(j)}, \fl^{(j+1)}$ with the same second coordinate, then we order
them in increasing order of the first coordinate. 
Let $K$ be the smallest positive integer so that 
\begin{equation}
\label{eq:defK}
\fl^{(K)}_1 = 0.
\end{equation}
By condition
(H1), $K$ exists. Now we say that the the lattice point 
$\fl \in \partial D_L$ is of type $k$ with $k=1,...,K$ if there 
exists an integer $m$ so that $\fl = \fl^{(mK+k)}$.

\begin{enumerate}
\item[(H2)]
{\bf Conditional local invariance principle}

{\it 
There are constants
$c_1,...,c_K$ so that
for any $0 < \alpha < \beta$ and for any 
$0 < \eta, \gamma < \xi$ the following holds. If 
$\fl \in \partial D_L$
is of type $k$, and $\fl_2 = \eta \sqrt T$, then

$$
\lim_{T \to \infty}  T^{3/2}
\prob \left( 
\cZ_T = \langle (\alpha, \gamma) \sqrt T \rangle,
\min \{
 \tau_0^{\cY},
\tau_{\xi \sqrt T }^{\cY}, 
\tau_0^{\cX},
\tau_{\beta \sqrt T }^{\cX} \} > T
| \cZ_0 = \fl \right)
$$
$$
= c_k \psi(\alpha, \beta) \phi(\eta, \gamma, \xi).
$$
Furthermore, for any $\eps >0$, the convergence is uniform 
for $\eps < \alpha < \alpha + \eps < \beta < 1/\eps$ and
$\eps < \eta < \eta + \eps < \xi < 1/\eps$,
$\eps < \gamma < \gamma + \eps < \xi$.
}

\item[(H3)] {\bf Moderate deviation bounds}
For any $x \in (0,1)$ and $y \in (-1,1)$, and for any $\fl = \fl^{(0)}, ...,  \fl^{(K-1)}$
$$
\lim_{\delta \to 0} \lim_{L \to \infty} \int_{[0, \delta L^2] \cup [L^2/\delta, \infty)}
L \Prob(
\cZ_t = \langle (xL,yL)  \rangle,
\min \{ 
\tau_0^{\cX},
\tau_{L }^{\cX} \} > t
| \cZ_0 = \fl )dt = 0
$$

\end{enumerate}


\subsection{Local equilibrium}

Consider now the Dirichlet problems
\begin{equation}
\label{eq:laplace2}
\Delta \tilde u = 0, \tilde u|_{\partial D} =  \tilde F,
\end{equation}
\begin{equation}
\label{eq:heat2}
\tilde v_t = \frac12 \left[ \tilde v_{xx} + \tilde v_{yy} \right], \quad
\tilde v(t,x,y)|_{(x,y) \in \partial D} =  \tilde F,
\tilde v(0,x,y) = 0.
\end{equation}
Here $\tF$ is defined by
$\tilde F: \partial D \to \reals$, 
$$
\tF(z) = 
\begin{cases}
\varsigma_W f_W(y) \text{ if }z = (0,y)\\
\varsigma_S f_S(x) \text{ if }z = (x,0)\\
\varsigma_E f_E(y) \text{ if }z = (A,y)\\
\varsigma_N f_N(y) \text{ if }z = (x,1),\\
\end{cases}
$$

where $f_E, f_W : [0,1] \to \reals$, $f_N, f_S : [0,A] \to \reals$
are given non-negative 
continuous functions and $\varsigma_{W/S/E/N}$ are non-negative real numbers
($W,S,E, N$ stand for West, South, 
North and East).
We perform the same procedure of injecting particles and absorbing
them on the boundary as before, but now we inject from all 4 sides of
the rectangle. Let $\tilde \Lambda_t$ denote the resulting measure 
defined as $\Lambda_t$.

We say that $\cZ$ satisfies that 
{\bf local equilibrium} (LE)
if 
for any $ t \in \reals_+ \cup \{ \infty\}$, for any
$k \in \mathbb Z_+$, for any $z_1,...,z_k$ distinct points in the interior $ D$
and for any distinct lattice points $\fl_1, ..., \fl_k \in \cL$, the joint distribution of 
$$
\mathfrak W_{t,i,j,L} := \tilde \Lambda_{t L^2}  (\langle z_i L \rangle + \fl_j), \quad  i,j=1,...,k
$$
converge weakly as $L \to \infty$ to independent Poisson random variables 
$\mathfrak W_{t,i,j,\infty}$
with expectation $\tilde v(t,z_i)$
(or 
$\tilde u(z_i)$ in case $t = \infty$),
where $\tilde v$ is defined by \eqref{eq:heat2} 
(and $\tilde u$ is defined by \eqref{eq:laplace2})
with some constants 
$\varsigma_{W/S/E/N}$. The points $\langle z_i L \rangle + l_j$, 
$j=1,...,k$
can be thought of as lying in a microscopic region near $\langle z_i L \rangle$.  
In particular, each point $\langle z_i L \rangle + l_j$ is a finite distance from $\langle z_i L \rangle$ so that it is in a "local" region of $z_i$ as $L$ becomes large. Indeed, the term {\it local equilibrium} refers to the fact that the limiting
distribution does not depend on $j$. We call the case $t \in \reals_+$
{\it local equilibrium in the hydrodynamic
limit} and the case $t = \infty$ {\it local equilibrium in the hydrostatic limit}. Since
in our case both hold at the same time, we simply refer
to these properties as local equilibrium.

Finally, we say that a lattice $\cL$ is 
{\bf rational}
if there are non-zero lattice points
$\fl^{(K_1),1}, \fl^{(K_2),2}$ in $\cL$ so that 
$\fl^{(K_1),1}_1 = \fl^{(K_2),2}_2 = 0$. Without loss of generality,
we assume that $\fl^{(K_1),1}_2 > 0$ and  $\fl^{(K_1),1}_2$ is the smallest among
such vectors with respect to the ordering introduced right after (H1) (and likewise 
for $\fl^{(K_2),2}$, except that in the ordering, the role of the first and second coordinates
are swapped). Clearly, if $\cL$ is rational, then (H1) holds with $K=K_1$ (and likewise,
a variant of (H1), where the two coordinates are swapped, holds with $K=K_2$).

Next, we show some examples, where we can verify conditions 
(H1) - (H3) and also prove (LE).

\section{Basic examples}
\label{sec3}
\subsection{Random walks}
\label{sec:rw}

Let $\tilde \cL \subset \reals^2$ be a 2 dimensional lattice. Let $\tilde \cP$ be a finitely
supported probability measure
on $\tilde \cL$ with zero expectation. We assume that there are finitely
many lattice points $\tilde w_1,...,\tilde w_J$ so that $\tilde \cP(\tilde w_j) > 0$ and 
$\sum \tilde \cP(\tilde w_j) = 1$. To avoid degeneracy, we assume that the group generated by $\tilde w_j$'s is $\tilde \cL$.

Let $\tilde \cZ$ be a homogeneous Markov
process: at exponential distributed times, $\tilde \cZ$ jumps with a 
jump distribution given by $\tilde \cP$. That is, the generator $\tilde G$ of
$\tilde \cZ$ is defined by
\begin{equation}
\label{eq:generator}
(\tilde Gf)(\fl) = \sum_{j=1}^J \tilde \cP
(\tilde w_j) [f(\tilde w_j + \fl) - f(\fl) ]
\end{equation}
for test functions $f: \tilde \cL \to \reals$.
By the central limit theorem, $\tilde \cZ_t / \sqrt t$ converges weakly to a
Gaussian distribution with mean zero and some 
covariance matrix $\Sigma$. Furthermore,
the non-degeneracy assumption ensures that $\Sigma$ is positive definite. 
Now we define $\cL =  \Sigma^{-1/2}\tilde \cL$, $w_j = \Sigma^{-1/2} \tilde w_j$,
$\cP (w_j) = \tilde \cP (\tilde w_j)$, $\cZ = \Sigma^{-1/2} \tilde \cZ$. 
\begin{proposition}
\label{thm:rw}
If $\cL$ is a rational
lattice, then in the above model (H1) - (H3) hold. 
\end{proposition}

We do not give a proof of Proposition \ref{thm:rw} as it follows
from a much simplified version of our proof of Theorem 
\ref{thm2}. In fact, the one dimensional version of (H2)
and (H3) is known 
for random walks, see \cite{C05,CC08}. We find it likely that
the two dimensional version is also known but we could not
find a reference. 

\subsection{Lorentz gas}
\label{sec:bil}

\subsubsection{Definitions}


We start with the definition of Sinai billiards \cite{S70}.
Consider a finite collection of strictly convex disjoint subsets $B_1,...,B_{\mathfrak k}$ of
the 2-torus
with $C^3$ boundary. The complement of these sets is denoted by
$\mathcal D_0 = \mathbb T^2 \setminus \cup_{i=1}^{\mathfrak k} B_i$
and is called the {\it configuration space}.
A point particle flies with constant speed inside 
$\mathcal D_0$ and undergoes specular reflection upon reaching $\partial \cD_0$
(i.e. angle of incidence equals the angle of reflection). Since the speed is
conserved, we obtain a continuous time dynamical system $\Phi_0^t$, $t \in \mathbb R$
on the {\it phase space} $\Omega_0 =  \cD_0 \times \mathcal S^1$.
The Sinai billiard flow $\Phi_0$ preserves the Lebesgue measure on $\Omega_0$
(denoted by $\mu_0$). We assume the finite horizon condition, i.e. that the sets 
$B_i$ are chosen in such a way that the free flight time is bounded.
Similarly, we define the periodic Lorentz gas when the phase space
is lifted to the universal cover. That is, the configuration space is
$\cD = \mathbb R^2 \setminus \cup_{(m,n) \in \mathbb Z^2} 
\cup_{i=1}^{\mathfrak k} (B_i + (m,n))$,
where we identify $\cD_0$ with $\cD \cap [-1/2,1/2)^2$. 
We choose this identification in such a way that
$(-1/2, -1/2) \notin \cD$.
The phase space is $\Omega = \cD \times \cS^1$ and the billiard flow is denoted by $\Phi^t$. It preserves
the $\sigma$-finite measure $\mu$, which is $\mu_0$ times the counting measure on $\mathbb Z^2$.

Now we construct the stochastic process which is the projection of the billiard flow, $\Phi^t$ onto $\mathbb Z^2$.  
Given $(q,v) \in \Omega$, let $\Pi_{\integers^2}(q,v) = (k,l) \in \mathbb Z^2$ 
if $q \in (k,l) + [-1/2,1/2)^2$ and let 
$\Pi_{\cD_0}(q,v) = q_0$,
$\Pi_{\Omega_0}(q,v) = (q_0,v)$ if $q = q_0 + \Pi_{\integers^2}(q,v)$.
We also put $\tilde \cZ_t(q,v) = \Pi_{\integers^2}(\Phi^t(q,v))$. Thus any probability measure on $\cD_0$ 
induces a stochastic process $\tilde \cZ_t$. 
It is important to note that here the
randomness only appears in the 
initial condition. Once $(q,v)$ is fixed, then $\tilde \cZ_t$ is uniquely defined for every $t$. 

We will also need the billiard map $\mathcal F_0$, which is defined as the Poincar\'e
section corresponding to the collisions, that is $\cF_0: \cM_0 \to \cM_0$, where
$$
\cM_0 = \{ (q,v) \in \partial \mathcal D_0 \times \cS^1: \langle v,n \rangle \geq 0\},
$$
where $n$ is normal to $\partial \cD_0$ at $q$ pointing inside $\cD_0$. 
The phase space
of the billiard map, $\cM_0$, thus corresponds to collisions where by convention we
use the post-collisional velocity $v$. $\cF_0$ preserves the probability 
measure $\nu_0$ defined by $d \nu_0 = c \cos \phi dr d \phi$, where $(r,\phi)$
are coordinates on $\cM_0$: $r$ is arclength parameter and $\phi \in [-\pi/2, \pi/2]$
is the angle between $v$ and $n$. The definitions of $\cM, \cF, \nu$ are analogous.

Fix a measure given by an arbitrary proper standard family 
(The exact definition will be given in Section \ref{sec:billb}. One example is the 
invariant measure $\nu$).
This measure induces a
stochastic process $\tilde \cZ_t$. Furthermore, $\tilde \cZ_t$ 
satisfies the central limit theorem with a 
covariance matrix which is independent of the standard family. 
That is, there exists a positive definite $2\times 2$
matrix $\Sigma$ so that $\tilde \cZ_T /\sqrt T $ converges weakly 
as $T \to \infty$ to the Gaussian distribution with 
mean zero and covariance matrix 
$\Sigma$ (see e.g. 
\cite{BSC91}). 
Now let
$\cL = \Sigma^{-1/2} \integers^2$, $\cZ_t = \Sigma^{- 1/2} \tilde \cZ_t$.
The {\it invariance principle} holds as well. That
is, $\cZ_{tT} /\sqrt T$, $t \in [0,1]$ converges weakly
to a standard Brownian motion (see e.g. \cite{C06}).

Recalling that  $(-1/2, 1/2) \notin \cD$, the graph $\bG$ 
on $\cL$ induced
by $\cZ$ satisfies that $(\fl, \fl')$
is an edge in $\bG$ if and only if $\Sigma^{1/2} \fl$
and $\Sigma^{1/2} \fl'$
are nearest neighbors in $\integers^2$. We will assume this in the sequel.

\begin{theorem}
\label{thm2} In the setting described above, assume that $\cL$ is a rational lattice. Then (H1)-(H3) hold.
\end{theorem}

Theorem \ref{thm2} does not claim that $\varsigma > 0$. In fact, there
are standard families for which $\varsigma = 0$. This is not surprising 
since in $\cZ_t$ can be deterministic for a bounded time. In particular,
we can choose a standard family so that $\cZ_t < 0 $ almost surely
for a fixed $t$ and so all particles will be absorbed within bounded amount
of time. However, there are standard families for which $\varsigma >0$
(e.g., for the invariant measure $\nu$). In general, we cannot compute
$\varsigma$ even if it is positive. 

Note that we assumed that $\cL$ is a rational lattice, which immediately gives (H1)
and the variant of (H1) when the vertical and the horizontal coordinates are swapped. 
This is a highly non trivial assumption and we expect this not to
hold for a typical billiard table. 
However, we have some examples when it does hold
due to some extra symmetry. We discuss these examples in Section \ref{sec:sym}. The proof of (H2) and (H3) will be given in Section
\ref{sec6}.

In case of deterministic systems like the Lorentz gas, a natural extension of (LE) is a finer
counting problem: that is to only count particles in a given nice subset of $\Omega_0$ (for
example, those that are close to a given scatterer). 
Let us fix and open set 
\begin{equation}
\label{def:A}
A \subset \Omega_0 \text{ with } \mu_0 (\partial A) = 0
\end{equation}
and update the definition of $\tilde \Lambda_t$ so as we only count particles at phase $(q,v)$ that
satisfies
$\Pi_{\Omega_0}(q,v) \in A$. Let the resulting measure be $\tilde \Lambda^A_t$ and let us say
that
detailed local equilibrium (DLE) holds if there is some $\varsigma$
so that for every $A$ as in \eqref{def:A}, the definition of (LE) with $\tilde \Lambda$ replaced
by $\tilde \Lambda_A$ holds with the constant $\varsigma \mu_0(A)$.

\begin{theorem}
\label{Cor1}
 Under the assumptions of Theorem \ref{thm2}, (LE) and (DLE) hold.
\end{theorem}

\begin{proof} [Proof of Theorem \ref{Cor1} assuming Theorem \ref{thm2}]

As observed in \cite{DN16}, the derivation of (LE) from \eqref{eq:thm11}
and \eqref{eq:thm12} is straightforward. 
Let $\mathbb M = (-tL^2, 0) \times \partial D_L \times \Omega_0$.
Let $\mathbb G : \mathbb M \to D_L \times \Omega_0 \cup\{ \infty \}$, where
$\mathbb G (s, \fl, (q,v)) = \Phi^s(q + \Sigma^{1/2} \fl, v)$
if the particle has not been absorbed by time $s$ and 
$\mathbb G (s, \fl, (q,v)) = \infty$ otherwise.
Since the initial conditions of particles is given by a Poisson point process
(PPP)
on $\mathbb M$, the mapping and restriction theorems for PPP (see e.g. sections 2.2 and 2.3 in \cite{Ki93}) give that
$\{ \mathbb G(s_i, \fl_i, x_i)\}_{\mathbb G(s_i, \fl_i, x_i) \neq \infty}$
forms a PPP on 
$D_L \times \Omega_0$. Letting $L \to \infty$, the intensity
measure of this PPP converges by Theorem \ref{thm1} and Theorem \ref{thm2} (for particles injected on the West or on the East this is immediate.
For particles injected on the North or the South, this follows from the
variant of Theorem \ref{thm1} when the role of $x$ and $y$ are swapped.
Since $\cL$ is assumed to be rational, (H1) holds even in this case). 
Thus in the limit $L \to \infty$, we obtain a PPP with intensity measure 
as on the right hand side of \eqref{eq:thm11}
and \eqref{eq:thm12}. This implies (LE). The proof of (DLE) is
analogous, 
except that when we verify (H2), we only need to take into account
particles at phase
$(q,v)$ that
satisfies
$\Pi_{\Omega_0}(q,v) \in A$. This requires a very minor change in the proof
(see the remark after Theorem \ref{thm:MLLT}). 
\end{proof}

\subsubsection{Symmetry Conditions }
\label{sec:sym}
\begin{example}
Assume that $\mathcal D_0$ is invariant under a $90$ degree rotation
or a vertical or horizontal reflection of the unit square. 
Then $\cL$ is rational.
\end{example}

\begin{proof}
Let us assume that $\mathcal D_0$ is invariant under a rotation by $90$ degrees. Then the probability
density function (pdf) of the limiting distribution of $\tilde \cZ_t / \sqrt t$ also needs to be invariant
under the rotation by $90$ degrees. Since this is a normal distribution, the isocontours of the pdf are
ellipses. The only ellipses invariant under the rotation by $90$ degrees are circles. This means that there is a
positive real number $\sigma$ so that $\Sigma = \sigma^2 I_2$.
Similarly, if $\mathcal D_0$ is invariant under reflection to vertical or horizontal axis, 
then the isocontours of limiting normal distribution are ellipses with semi axes parallel to the coordinate
axes and so $\Sigma$ is diagonal.
\end{proof}
In the above examples, $\cL$ is generated by $\sigma_1^{-1}[1,0]^T$ and 
$\sigma_2^{-1}[0,1]^T$. Consequently, $K_1=K_2=1$. In this
sense, these examples are the simplest possible ones
(Figure \ref{fig2} shows a configuration which symmetric with respect to the vertical axis,
and is repeated over a $10 \times 6$ rectangle).
Our next example is less trivial as
$K_2=2$.

\begin{example}
Consider
a scatterer configuration on the regular hexagon that is invariant under the rotation by $120$ degrees
and satisfy all other assumptions (that is, the scatterers are smooth, disjoint, strictly convex
and the configuration has finite horizon). One such example is only one scatterer which is a disc,
centered at the center of the hexagon and with radius large enough to ensure that 
$\mathcal D$ is of finite horizon.
By tiling the plane with regular hexagons, we obtain the Lorentz gas as before. As in 
the previous example, the isocontours of the limiting normal distribution are invariant under the 
rotation by $120$ degrees; hence they are circles and $\Sigma = \sigma^2 I_2$. In this case, $\tilde{\cZ}_t$
for any $t$ takes values in the set of tiles of the hexagonal tiling. Let ${\cL}$ be the lattice generated
by the vectors $\sigma^{-1}[0,1]^T$ and $\sigma^{-1}[ \sqrt 3 /2, 1/2]^T$ and $\bG$ be the graph with 
vertices $\cL$ and edges between points at distance $\sigma^{-1}$. That is, $\bG$ forms the triangular
grid, dual to the hexagonal tiling (see Figure \ref{fig:hex}, the edges of $\bG$ are denoted by dotted lines). 
In this example, $K_1=1$ and $K_2=2$. Indeed, on the 
horizontal boundary, we see an alternating sequence of two kinds of hexagons (ignoring the very first and the 
very last one): one of them has $5$ neighbors in $D_L$ and the other one only has $3$. 
A particle injected
to a uniform random location on the first type hexagon has higher chance of staying in $D_L$
than in case of the second type hexagon. 
Thus we expect that $c_1 \neq c_2$ in the variant of (H2), when the 
vertical and horizontal coordinates are swapped.
\end{example}

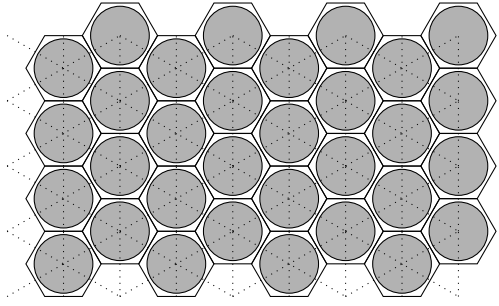
\begin{figure}
\begin{center}
\begin{tikzpicture}[x=7.5mm,y=4.34mm]
  \tikzset{
    box/.style={
      regular polygon,
      regular polygon sides=6,
      minimum size=10mm,
      inner sep=0mm,
      outer sep=0mm,
      rotate=0,
    draw
    }
  }

\foreach \i in {0,...,3} 
    \foreach \j in {0,...,3} {
            \node[box] at (2*\i,2*\j) {};
            \node[box] at (2*\i+1,2*\j+1) {};
            \draw[fill=gray!60!white] (2*\i,2*\j) circle (0.39 cm);
            \draw[fill=gray!60!white] (2*\i+1,2*\j+1) circle (0.39 cm);
            \draw[dotted] (2*\i,2*\j-1) -- (2*\i,2*\j+1);
            \draw[dotted] (2*\i-1,2*\j-1) -- (2*\i+1,2*\j+1);
            \draw[dotted] (2*\i-1,2*\j+1) -- (2*\i+1,2*\j-1);
            \draw[dotted] (2*\i+1,2*\j-1) -- (2*\i+1,2*\j+1);
       }

\end{tikzpicture}
\caption{Billiard configuration on a hexagonal tiling} \label{fig:hex}
\end{center}
\end{figure}








%


\section{Duality}
\label{sec:dual}

\subsection{Random walks}

The definitions given in Section \ref{sec:def1} easily extend to more 
general domains $D$ with piece-wise smooth boundary.
One minor difference is that in \eqref{eq:Poi} instead of $f(\fl_2/L)$
we need to choose a slightly different argument of $f$ as 
$\fl_2/L$ may not be on $\partial D$ and $f$ may not be defined 
(e.g. one can choose the closest point on $\partial D$ to $\fl /L$).
Since $f$ is continuous, the exact choice is irrelevant as long as it
is a bounded distance from $\fl/L$. To keep notations simple, 
we will write $f(\fl/L)$, where $f$ is a continuous function defined on 
$\partial D$ (there is no need to introduce $F$). 

\begin{proposition}
\label{prop:RW}
Consider a random walk as in Section \ref{sec:rw} and 
let
\begin{equation}
\label{eq:specA}
\bA(\cJ) = \sum_{j \in \cJ} \cP (w_j)
\end{equation}
and $\bB = 1$. 
Then the conclusion of Theorem \ref{thm1} and (LE) hold 
with $\varsigma =1$ without assuming the 
rationality of $\cL$ and for general
bounded domains $D$ with piece-wise smooth boundary and no cusps.
\end{proposition}

\begin{proof}
We are only going to prove
\eqref{eq:thm12}. A proof of \eqref{eq:thm11} can be obtained
by replacing $t$ by $\infty$ in the proof below and
(LE) can be proved as in Theorem \ref{Cor1}.

The key idea of the proof is duality.
Specifically, we use the fact that the reversed
Markov process is also Markovian. 
Let $\check{\cZ}$ be the discretized version of $\cZ$. That is,
$\check{\cZ}_0 = \cZ_0$, $\check{\cZ}_n = \cZ_{t_n}$ where $t_n$ is the 
time of the $n$th jump of $\cZ$. 
The reverted random walk $\cZ'$ is defined by
the generator
$$
 (G' f)(\fl) = \sum_{j=1}^J \cP
(w_j) [f(-w_j + \fl) - f(\fl) ]
$$
and $\check{\cZ}'$ is the discretized version of $\cZ'$ (defined analogously to $\check{\cZ}$).

Note that for any $N$, $\check{\cZ}$ induces a measure 
$\prob_{\check{\cZ}}$ on 
$\cL^N$ by 
$$
\prob_{\check{\cZ}} (\fl_0,...,\fl_{N-1})
= 
\prob (\check{\cZ}_1 = \fl_1, ..., \check{\cZ}_{N-1} = \fl_{N-1} |
\check{\cZ}_0 = \fl_0).
$$
Let us define $\prob_{\check{\cZ}'}$ analogously. Then by definition
of $\check{\cZ}'$, for any sequence $\fl_0,...,\fl_{M} \in \cL$,
\begin{equation}
\label{eq:dual1}
\prob_{\check{\cZ}} (\fl_0,...,\fl_{M}) = 
\prob_{\check{\cZ}'} (\fl_M,...,\fl_{0}) .
\end{equation}
For fixed $L$, $z \in D$, $t \in \mathbb R_+$, $\fl \in \partial D_L$ and $M$, let
$\cA = \cA_{L,z,t,\fl,M}$
be the set of length $M$ trajectories from $\fl$ to $\langle zL \rangle$ staying
inside $D_L$, i.e.
$$
\cA =
\{ 
(\fl_0,...,\fl_M): \fl_0 = \fl, \forall i=0,...,M-1: \exists j=1,...,J:
\fl_{i+1} - \fl_i = w_j, \fl_i \in D_L, \fl_M = \langle zL \rangle
\} .
$$
For a subset $\cB \subset \cL^{M+1}$ and a lattice point $\hat{\fl}$, let 
$$
\cB' = \{ 
(\fl_M,...,\fl_0): (\fl_0,...,\fl_M) \in \cB
\},
$$
and
$$
\hat{\fl} \cB = \{ (\hat{\fl},\fl_0,...,\fl_M): (\fl_0,...,\fl_M) \in \cB\},\quad
 \cB \hat{\fl}= \{ (\fl_0,...,\fl_M,\hat{\fl}): (\fl_0,...,\fl_M) \in \cB\}.
$$
Then by \eqref{eq:dual1}, we have
$$
\prob_{\check{\cZ}} (\cA_{L,z,t,\fl,M}) = \prob_{\check{\cZ}'} (\cA'_{L,z,t,\fl,M}),
$$
Furthermore, for any $\fl_{-1} \notin D_L$, which is connected to $\fl$ in $\bG$,
\begin{equation}
\label{eq:dual2}
\prob_{\check{\cZ}} (\fl_{-1}\cA_{L,z,t,\fl,M}) = \prob_{\check{\cZ}'} (\cA'_{L,z,t,\fl,M} \fl_{-1}). 
\end{equation}
Let $\bT$ be the first hitting time of $\cL \setminus D_L$ by
$\check{\cZ}'$. Then \eqref{eq:dual2} is equal to 
$$
\prob (\bT = M+1, \check{\cZ}'_{M+1} = \fl_{-1}, \check{\cZ}'_{M} = \fl
| \check{\cZ}'_0 = \langle zL \rangle ).
$$
To turn to continuous time, let
${\tau'}^{*}$ is the first time $\cZ'$ leaves $D_L$. Then we have
\begin{equation}
\label{eq:dual2.5}
\prob( {\tau'}^{*} < tL^2,
\cZ'_{{\tau'}^{*}} = \fl_{-1}, \cZ'_{{\tau'}^{*}-} = \fl
|\cZ'_0 =  \langle zL \rangle
) = \sum_{M=0}^{\infty} F_{M+1}(tL^2)
\prob_{\check{\cZ}} (\fl_{-1}\cA_{L,z,t,\fl,M}),
\end{equation}
where $F_{N} (.)$ is the cumulative distribution function of the Gamma
distribution with shape parameter $N$
and scale parameter $1$ (that is, it is the sum of $N$ iid exponential random
variables, each with expectation $1$). Indeed, \eqref{eq:dual2.5} holds since
the time of jumps of the Markov process $\cZ'$ are independent of the location
of the jump.
On the other hand, we have
\begin{align}
&\sum_{M=0}^{\infty} F_{M+1}(tL^2)
\prob_{\check{\cZ}} (\fl_{-1}\cA_{L,z,t,\fl,M})
= 
\cP(\fl - \fl_{-1}) \sum_{M=0}^{\infty} F_{M+1}(tL^2)
\prob_{\check{\cZ}} (\cA_{L,z,t,\fl,M}) \nonumber \\
&= \cP(\fl - \fl_{-1}) \int_{0}^{tL^2}
\Prob (\cZ_{s} = \langle zL \rangle, \forall s' \in [0,s],
\cZ_{s'} \in D_L | \cZ_0 = \fl) ds. \label{eq:dual3}
\end{align}
Since $\bB = 1$, we have
\begin{equation}
\label{eq:dual4}
\Lambda_{tL^2}
(\langle zL \rangle) = 
\sum_{\fl \in \partial D_L} 
\bA(\cJ(\fl)) f(\fl/ L) \int_{0}^{tL^2} \Prob (\cZ_{s} = \langle zL \rangle, \forall s' \in [0,s],
\cZ_{s'} \in D_L | \cZ_0 = \fl) ds.
\end{equation}
Thus by \eqref{eq:specA} and \eqref{eq:dual3}, we have
$$
\Lambda_{tL^2}
(\langle zL \rangle) = 
\sum_{\fl: \in \partial D_L} 
\sum_{\fl_{-1} \in \cL \setminus D_L: (\fl_{-1}, \fl) \in \bG} 
f(\fl/ L) \sum_{M=0}^{\infty} F_{M+1}(tL^2)
\prob_{\check{\cZ}} (\fl_{-1}\cA_{L,z,t,\fl,M})
$$
and so by 
\eqref{eq:dual2.5},
\begin{equation}
\label{eq:dual5}
\Lambda_{tL^2} 
(\langle zL \rangle) =  \mathbb E \left(f 
\left(\frac{ \cZ'_{{\tau'}^{*}-}}{L} \right) 1_{{\tau'}^{*} < tL^2 } 
| {\cZ}'_0 = \langle zL \rangle
\right).
\end{equation}
Now the right hand side of \eqref{eq:dual5} converges, as $L \to \infty$ to
\begin{equation}
\label{eq:dual6}
\mathbb E \left(f 
\left( W_{\cT^*} \right) 1_{\cT^*< t} 
| W_0 = z
\right),
\end{equation}
where $W_t$ is a standard planar Brownian motion and $\cT^*$ is the hitting time of 
$\reals^2 \setminus D$ by $W$.
(This follows from Donsker's theorem and the continuous mapping theorem. 
A more detailed proof of \eqref{eq:dual6}
for the case $t = \infty$ can be found in
e.g. \cite[Proposition 3]{LNY}.)
Let $\cW$ be a diffusion
process whose first coordinate is deterministic with constant $1$ drift
and whose second and third coordinates are 
independent standard Brownian motions. 
Applying Dynkin's formula for $\cW$ with $\cW_0 = (-t, z)$, 
the stopping time $\bt$ as the first hitting time of 
$\reals^3 \setminus ([-t,0] \times D )$, and 
with the test function $v(-s,\tilde z)$, where
$v$ is defined by \eqref{eq:heat}, we conclude that
 \eqref{eq:dual6} satisfies \eqref{eq:thm12} with $\varsigma = 1$.

\end{proof}

We record a remark for later reference:

\begin{remark}
\label{rem}
Note that the proof of Proposition \ref{prop:RW} does not use Theorem \ref{thm1}.
Thus we already have an example (random walks), where
both the assumptions
and the conclusion
of Theorem \ref{thm1} are verified (by Proposition \ref{thm:rw} and Proposition \ref{prop:RW},
respectively).
\end{remark}

\subsection{Lorentz gas}

Let $D$ be a bounded domain with piece-wise smooth boundary and 
no cusps. In the setup of Section \ref{sec:bil},
given $D$, $L$, and
$\fl \in \partial D_L$, we consider the following
initial measure. For any $\fl_- \in \cL \setminus D_L$ connected to $\fl$
in $\bG$, $\tilde{\fl'} := \Sigma^{1/2} \fl'$ is a nearest neighbor of 
$\tilde{\fl} := \Sigma^{1/2} \fl$ in $\integers^2$
(that is, 
$\tilde{\fl} - \tilde{\fl'} \in \{ w_1 = (0,-1), w_2 = (0,1),w_3 = (-1,0),w_4 = (1,0)\})$ 
by our assumption in Section \ref{sec:bil}). Let $E 
= E_{\fl,\fl'} \subset \reals^2$ be the
line segment on the boundary of $\tilde{\fl} + [-1/2, 1/2)^2$ and
$\tilde{\fl'} + [-1/2, 1/2)^2$.
Define
$$
\cN = \cN_{\fl,\fl'}= \{ 
(q,v) \in \Omega: q \in E, \langle v,\tilde{\fl} - \tilde{\fl'}\rangle > 0
\}.
$$
Let $\mbox{type}(\fl, \fl') = j$ if
$\tilde{\fl} - \tilde{\fl'} = w_j$
and $\zeta_{j}: \cN_{0,\Sigma^{-1/2}w_j} \to \reals_+$ be the first return to $\cN_{0,\Sigma^{-1/2}w_j}$
in the compact Sinai billiard, that is 
$$
\zeta_j = \min \{ s: \Phi_0^s(q,v) \in \cN_{0,\Sigma^{-1/2}w_j} \}.
$$
Let us also write 
$$\bar \zeta_j = \int_{\cN_{0,\Sigma^{-1/2}w_j}}
\zeta_j d \varrho_{0,\Sigma^{-1/2}w_j}.
$$
Next, we define the finite measure 
$\varrho = \varrho_{\fl,\fl'}$ on $\cN$ by 
$$
d\varrho = \frac{1}{2 \bar \zeta_j} \cos(\langle v,\tilde{\fl} - \tilde{\fl'}\rangle) dq dv,$$
where $\mbox{type}(\fl, \fl') = j$. Note that 
$\varrho(\cN) = |E_{\fl,\fl - w_j}| / \bar \zeta_j$ 
and so it may not be a probability
measure.
Now 
the initial condition $\cG$ is given by the normalized 
sum of these measures for all
neighbors $\fl'$. That is, 
$$\nu_{\cG} = \frac{1}{\sum_{j \in \cJ(\fl)} 
\frac{|E_{\fl,\fl - w_j}|}{\zeta_j}}
\sum_{j \in \cJ(\fl)} 
\varrho_{\fl,\fl - w_j}.$$
By definition, $\nu_{\cG}$ is a probability measure. 
Next, we choose
$$
\bA(\cJ(\fl)) = \sum_{j \in \cJ(\fl)} \frac{|E_{\fl,\fl - w_j}|}{\zeta_j}
$$
(which clearly depends on $\fl$ only through $\cJ(\fl)$)
and $\bB = 1$. This choice guarantees that particles are being continuously
injected through the entire boundary of $D_L$ with a 
measure which is simply the projection of the invariant measure $\mu$
to the Poincar\'e section on the boundary of $D_L$.
Because of this very special choice of $\nu_{\cG}, \bA, \bB$, we have

\begin{proposition}
\label{prop:dualbil}
With the above choice, the conclusion of Theorem \ref{thm1},
(LE) and (DLE) hold with $\varsigma =1$ without assuming the 
rationality of $\cL$ and for general
bounded domains $D$ with piece-wise smooth boundary and no cusps.
\end{proposition}

\begin{proof}
The proof is similar to that of Proposition \ref{prop:RW}. 
We use duality and it is sufficient to verify
\eqref{eq:thm12}.

We claim that there is some $s^* > 0$
so that for any $(q,v) \in \cN_{\fl, \fl'}$ and any 
$s \in [0, s^*]$, 
$\tilde{\cZ}_s(q,v) \in \{ \tilde{\fl}, \tilde{\fl'}\}$. Furthermore, 
if there is some $s \in [0, s^*]$ with $\tilde{\cZ}_s(q,v) = \fl'$, then 
$\tilde{\cZ}_{s^*}(q,v) = \fl'$.
Indeed, 
the first statement follows from the assumption 
that $(-1/2, 1/2) \notin \cD$ 
and the second follows from the fact that visiting $\fl$, then $\fl'$ and
then $\fl$ again requires at least $2$ collisions and so we 
choose $s^*$ shorter than the minimal free flight.

Next, for any $(\fl, \fl')$ as above, by the definition of $\varrho$
and by the fact that $s^* < \min \zeta$, we have for measurable sets 
$B \subset \cup_{s \in [0, s^*]} \Phi^s(\cN_{\fl, \fl'})$
\begin{equation}
\label{eq:susp}
\int B d \mu = 
\int_{0}^{s^*}  \left( \int B d \Phi^s_*(
\varrho_{\fl, \fl'}) \right) ds.
\end{equation}
By the definition of $\nu_{\cG}, \bA$ and $\bB$, we have
\begin{align*}
\Lambda_{tL^2} (\langle zL \rangle) = &\sum_{\fl: \in \partial D_L} 
\sum_{\fl' \in \cL \setminus D_L: (\fl', \fl) \in \bG} f(\fl / L)\\
&
\int_{0}^{tL^2}
\int_{\cN_{\fl,\fl'}} \{
(q,v): \forall s' \in [0,s], \cZ_{s'} (q,v) \in D_L,
\cZ_{s} (q,v)  = \langle zL \rangle 
\} d\varrho_{\fl, \fl'}(q,v) ds.
\end{align*}

For fixed $t$ and $L$, let $K \in \integers_+$ so that $K s^* \leq tL^2 < (K+1)s^* $.
To simplify formulas, let us assume that 
$K s^* = tL^2$ holds (it is easy to check 
that the contribution of $s \in [K s^*, tL^2]$ is negligible). Now for
$k = 1,...,K$ we apply
\eqref{eq:susp} with
$$
B_{\fl, \fl',k} = \{ (q,v) \in \cup_{s \in [0, s^*]} \Phi^s(\cN_{\fl, \fl'}): 
\forall s' \in [0, (k-1)s^*], 
\cZ_{s'}(q,v) \in D_L,
\cZ_{(k-1)s^*}(q,v) = \langle zL \rangle \} 
$$
and the definition of $s^*$ to conclude
$$
\int_{(k-1)s^*}^{ks^*}
\int_{\cN_{\fl,\fl'}} \{
(q,v): \forall s' \in [0,s], \cZ_{s'} (q,v) \in D_L,
\cZ_{s} (q,v)  = \langle zL \rangle 
\} d\varrho_{\fl, \fl'}(q,v) ds
= \int B_{\fl, \fl',k}  d \mu
$$
and so
\begin{equation}
\label{eq:inv1}
\Lambda_{tL^2} (\langle zL \rangle) = \sum_{\fl: \in \partial D_L} 
\sum_{\fl' \in \cL \setminus D_L: (\fl', \fl) \in \bG} f(\fl / L)
 \sum_{k=1}^K
\int (B_k) d \mu.
\end{equation}

Now we recall the involution (also known as 
time reversibility) property of billiards. For $(q,v) \in \Omega$,
let $\cI(q,v) = (q, -v)$. Then $\cI$ preserves $\mu$ and anticommutes 
with the flow. That is,
$$
\Phi^{-s} \circ \cI = \cI \circ \Phi^s.
$$
(see e.g. \cite[Section 2.14]{ChM}). Thus 
\begin{equation}
\label{eq:inv2}
\int B_{\fl, \fl',k}  d \mu
= \int B'_{\fl, \fl',k}  d \mu,
\end{equation}
where
\begin{align}
B'_{\fl, \fl',k} = & \{ (q,v) \in \Omega : \cZ_0(q,v) = \langle zL \rangle \nonumber \\
& \exists s \in [(k-1)s^*, ks^*]: \forall s' \in [0, s]:
\cZ_{s'} \in D_L, \Pi_{\cD} \Phi^s(q,v) \in E_{\fl, \fl'}
\}  .\label{eq:inv3}
\end{align}
Using the notation \eqref{eq:tau*} and combining 
\eqref{eq:inv1}, \eqref{eq:inv2} and \eqref{eq:inv3}, we conclude
\begin{equation}
\label{eq:inv4}
\Lambda_{tL^2} (\langle zL \rangle) =
\int_{(q,v): \cZ_0(q,v) = \langle zL \rangle }
f \left( \frac{\cZ_{\tau^*-}}{L}\right) 1_{\tau^* < tL^2}
d \mu
\end{equation}
By the invariance principle, the right hand side of
\eqref{eq:inv4} converges as $L \to \infty$ to 
\eqref{eq:dual6}. As in 
Proposition \ref{prop:RW},
\eqref{eq:thm12} follows.

\end{proof}


\section{Proof of Theorem \ref{thm1}}
\label{sec5}

The keep the notations simpler, we assume that $a=1$ (the proof extends to any $a>0$ with no new
ideas). 
We will prove \eqref{eq:thm11} first.
Let $z = (x,y)$ be a point in the interior of $D$.
By definition, we have
\begin{align}
\mathbb E (\Lambda (\langle zL \rangle ))
=& \int_0^{\infty} \sum_{\fl \in \partial_W D_L} 
\bA(\cJ(\fl)) \bB(t) 
f\left(\frac{\fl_2}{L} \right) \nonumber \\
&\mathbb P \left(
\cZ_t = \langle (x, y) L \rangle,
\min \{
 \tau_0^{\cY},
\tau_{L}^{\cY}, 
\tau_0^{\cX},
\tau_{L}^{\cX} \} > t
| \cZ_0 = \fl \right) dt \nonumber\\
=&  
\int_{\delta L^2}^{L^2 / \delta} ... dt + \int_{0}^{\delta L^2} ... dt + 
\int_{L^2 / \delta}^{\infty} ... dt  = : I_1 + I_2 + I_3  \label{eq:dec}
\end{align}
with $I_j = I_j(L,x,y,\delta)$ for $j=1,2,3$.
Noting that
\begin{equation}
\label{eq:I1,3}
\lim_{\delta \to 0} \lim_{L \to \infty} I_2 + I_3 = 0.
\end{equation} by (H3), it remains
 to prove
\begin{equation}
\label{eq:I2}
\lim_{\delta \to 0} \lim_{L \to \infty} I_1 = u(z).
\end{equation}

Let $\Psi_{\delta'}:[0,1] \to [0,1]$ be defined by
$$
\Psi_{\delta'}(y) = 
\begin{cases}
0 & \text{ if } y < \delta'\\
\frac{1}{\delta'}y  - 1& \text{ if } \delta' \leq y < 2\delta'\\
1& \text{ if } 2\delta' \leq y < 1 - 2\delta'\\
-\frac{1}{\delta'}y  - 1 + \frac{1}{\delta'} &\text{ if } 1-2 \delta' \leq y < 1 -\delta'\\
0 & \text{ if } y > 1 -\delta'\\
\end{cases}
$$
and write $f_{\delta'}(y) = f(y) \Psi_{\delta'}(y)$.

To prove \eqref{eq:I2}, we first write 
$I_1 = I_{11} + I_{12}
$ with $I_{1,k} = I_{1,k}(L,x,y,\delta,\delta')$ for $k=1,2$, where 
$I_{11}$ and $I_{12}$ are obtained from $I_1$ by replacing $f$ by $f_{\delta'}$ and
$f-f_{\delta'}$, respectively.
To verify \eqref{eq:I2}, it is sufficient to prove
\begin{equation}
\label{eq:I11}
\lim_{\delta' \to 0} \lim_{\delta \to 0} \lim_{L \to \infty} I_{11} = u(z)
\end{equation}
and
\begin{equation}
\label{eq:I12}
\lim_{\delta' \to 0} \lim_{\delta \to 0} \lim_{L \to \infty} I_{12} = 0
\end{equation}
To simplify notations, we will write $I_{11}^{\infty} = \lim_{L \to \infty} I_{11}$
and $I^{\infty,0}_{11} :=
\lim_{\delta \to 0} I^{\infty}_{11}$.

Let us consider the following truncated version of \eqref{eq:laplace}
\begin{equation}
\label{eq:laplace3}
\Delta \hat u = 0, \quad \hat u|_{\partial D} = \varsigma F_{\delta'},
\end{equation}
where $F_{\delta'}$ is defined as $F$ except that $f$ is replaced by $f_{\delta'}$.

\begin{proposition}
\label{prop1}
For any $\delta' \in (0, 1/4)$, $I^{\infty,0}_{11}$ is the solution of \eqref{eq:laplace3}.
\end{proposition}

\begin{proof}
The proof consists of two steps. First, we prove that $I^{\infty,0}_{11}$ exists; then we show that
it satisfies \eqref{eq:laplace3}.

{\bf Step 1: $I^{\infty,0}_{11}$ exists}

Let us define $B = \fl^{(K)}_2$, where $K$ is defined by \eqref{eq:defK}. 
To simplify formulas, let us write
$\bar \tau = \min \{
 \tau_0^{\cY},
\tau_{L}^{\cY}, 
\tau_0^{\cX},
\tau_{L}^{\cX} \}$.
Also observe that by transitivity of $\bG$, there are constants
$\bA_1,...,\bA_K$ so that for any $m \in \mathbb N$, for any $k=1,...,K$,
$\bA(\cJ(\fl^{(mK+k)})) = \bA_k$.  
Now, we compute 
\begin{align*}
I_{11} =& \sum_{\fl \in \partial_W D_L, \fl_2/L \in (\delta', 1-\delta')}
\bA(\cJ(\fl)) 
\int_{\delta L^2}^{L^2 / \delta} 
\bB(t)
f_{\delta'}\left(
\frac{\fl_2}{L} \right) \mathbb P (
\cZ_t = \langle (x, y) L \rangle,
\bar \tau > t
| \cZ_0 = \fl) dt\\
=&
\sum_{m = \delta' L /B}^{(1-\delta' ) L /B} \sum_{k=1}^{K} \bA_k
\int_{\delta L^2}^{L^2 / \delta} \bB(t)
f_{\delta'}\left(\frac{\fl^{(mK+k)}_2}{L} \right)
 \mathbb P (
\cZ_t = \langle (x, y) L \rangle,
\bar \tau  > t
| \cZ_0 = \fl^{(mK+k)}) dt\\
=& 
\sum_{m = \delta' L /B}^{(1-\delta' ) L /B} \sum_{k=1}^{K} \bA_k\\
&
\int_{\delta}^{1/ \delta} \bB(s L^2) 
f_{\delta'}\left(\frac{\fl^{(mK+k)}_2}{L} \right)
 \mathbb P (
\cZ_{sL^2} = \langle (x, y) L \rangle,
\bar \tau  > sL^2
| \cZ_0 = \fl^{(mK+k)}) L^2ds
\end{align*}
Now using (H2) with $T = sL^2$,
$\alpha = x/\sqrt s$,
$\beta = 1/\sqrt s$,
$\eta = \fl^{(mK+k)}_2/(L \sqrt s)$,
$\gamma = y/\sqrt s$,
$\xi = 1/\sqrt s$, we obtain
\begin{align*}
I_{11} \sim&
\sum_{m = \delta' L /B}^{(1-\delta' ) L /B} \sum_{k=1}^{K}
\bA_k
c_k\\
&
\int_{\delta}^{1/ \delta} \bB(s L^2) 
f_{\delta'}\left(\frac{\fl^{(mK+k)}_2}{L} \right)
s^{-3/2} L^{-1}
\psi  \left( \frac{x}{\sqrt s}, \frac{1}{\sqrt s} \right)  
\phi
\left( \frac{\fl^{(mK+k)}_2}{L \sqrt s}
, \frac{y}{\sqrt s}, \frac{1}{\sqrt s} \right)  ds
\end{align*}
by uniform convergence, where $a_L \sim b_L$ means that 
$\lim_{L \to \infty} a_L/b_L = 1$.
Let us write 
$$\brc = \frac{1}{K} \sum_{k=1}^K \bA_k c_k .$$ 
Then
$$
I_{11} \sim
\frac{\brc K}{B}
\int_{\delta}^{1/ \delta}  \bB(s L^2) 
s^{-3/2}
\psi  \left( \frac{x}{\sqrt s}, \frac{1}{\sqrt s} \right)  
\left[
\sum_{m = \delta' L /B}^{(1-\delta' ) L /B} \frac{B}{L}
f_{\delta'}\left(\frac{\fl^{(mK)}_2}{L} \right)
\phi
\left( \frac{\fl^{(mK)}_2}{L \sqrt s}, 
\frac{y}{\sqrt s}, \frac{1}{\sqrt s} \right)  
\right]ds.
$$
Replacing the Riemann sum with the corresponding 
Riemann integral, we obtain
$$
I_{11} \sim
\frac{\brc K}{B}
\int_{\delta}^{1/ \delta} \bB(s L^2) 
s^{-3/2}
\psi  \left( \frac{x}{\sqrt s}, \frac{1}{\sqrt s} \right)  
\left[
\int_{\delta'}^{1-\delta'} f_{\delta'}(\sigma) 
\phi
\left( \frac{\sigma}{\sqrt s}, 
\frac{y}{\sqrt s}, \frac{1}{\sqrt s} \right)  d \sigma
\right]ds
$$
(we are permitted to do this because of uniform convergence of the bracketed expression in 
$s$).
Since the the integrand in the last formula 
is uniformly continuous in $s$ and since $\bB$ is periodic with period $1$ and
$\int_0^1 \bB = 1$, we can take the
limit $L \to \infty$
to conclude that $I_{11}^{\infty}$ exists and
is equal to 
$$
\frac{\brc K}{B}
\int_{\delta}^{1/ \delta}
s^{-3/2}
\psi  \left( \frac{x}{\sqrt s}, \frac{1}{\sqrt s} \right)  
\left[
\int_{\delta'}^{1-\delta'} f_{\delta'}(\sigma) 
\phi
\left( \frac{\sigma}{\sqrt s}, 
\frac{y}{\sqrt s}, \frac{1}{\sqrt s} \right)  d \sigma
\right]ds.
$$
Now we substitute \eqref{eq:phi} and \eqref{eq:psi} to 
the above to conclude
\begin{align*}
I_{11}^{\infty} &=
\frac{\brc K}{B}  \int_{\delta'}^{1-\delta'}  \int_{\delta}^{\frac{1}{\delta}}  \sum_{k=-\infty}^{\infty}  \sum_{n=-\infty}^{\infty}  \Bigg(
		 \frac{1}{s^2}  (2k + x)\exp \Big( -\frac{(2k+x)^2}{2s} \Big)  \frac{1}{\sqrt{2\pi}} \\
	& f_{\delta'}(\sigma) \Big[
		\exp \Big( -\frac{(y-\sigma-2n)^2}{2s} \Big) - \exp \Big( -\frac{(y+\sigma+2n)^2}{2s} \Big) \Big] \Bigg) ds d\sigma.
\end{align*}
Clearly, the sum is absolutely and uniformly convergent and so we can write the
sums in front of the integrals. Thus
$$
I_{11}^{\infty} = \frac{\brc K}{B} \sum_{k=-\infty}^{\infty}  \sum_{n=-\infty}^{\infty}
\int_{\delta'}^{1-\delta'} R(k, n, \delta, \sigma, s, x, y) d\sigma,
$$
where

\begin{align*}
&	R(k,n,\delta, \sigma, s, x, y) 
	=   \frac{x + 2k}{\sqrt{2\pi}} f_{\delta'}(\sigma) \\
	& * \int_{\delta}^{\frac{1}{\delta}}
		 \frac{1}{s^2}  \Big[ \exp \Big( -\frac{(2k+x)^2 + (y-\sigma-2n)^2}{2s} \Big)  
		 - \exp \Big( -\frac{(2k+x)^2 + (y+\sigma+2n)^2}{2s} \Big) \Big]  ds .
\end{align*}

Making the substitution $\omega = (2s)^{-\frac{1}{2}}$ 
(and so $4 \omega d \omega = - ds/s^2$)
and letting 
$P_1 = (2k+x)^2 + (y-\sigma-2n)^2$ and $P_2 =  (2k+x)^2 + (y+\sigma+2n)^2$, 
we get:
$$
	R(k, n, \delta, \sigma, x, y) = \frac{4(x + 2k)}{\sqrt{2\pi}} f_{\delta'}(\sigma) \int_{\sqrt{ \delta /2}}^{\frac{1}{\sqrt{2 \delta}}}\
		\omega \Bigg[ \exp (-P_1 \omega^2 )  
		 - \exp (-P_2 \omega^2  ) \Bigg] d\omega 
$$
$$ 
	=   -\frac{2(x + 2k)}{\sqrt{2\pi}} f_{\delta'}(\sigma) \Bigg[  \frac{1}{P_1}\exp\left( -\frac{P_1}{2\delta} \right) - \frac{1}{P_2}\exp\left( -\frac{P_2}{2\delta} \right) 
	   \Bigg] $$
$$+ \frac{2(x + 2k)}{\sqrt{2\pi}} f_{\delta'}(\sigma) \Bigg[ 
	   \frac{1}{P_1}\exp\left( -\frac{P_1 \delta}{2} \right)  - \frac{1}{P_2}\exp\left( -\frac{P_2 \delta}{2} \right) 
	   \Bigg] 
$$
$$ 
	= : R_1+R_2.
$$
Clearly, we have 
$$
\lim_{\delta \rightarrow 0} \sum_n \sum_k R_1 = 0
$$
and as Lemma \ref{lem:5.3} shows,
$$
\lim_{\delta \rightarrow 0} \sum_n \sum_k R_2 = \sum_n \sum_k \lim_{\delta \rightarrow 0} R_2 .
$$
So we get
$$
	 \lim_{\delta \to 0} R(k, n, \delta, \sigma, x, y) = R(k,n, \sigma, x, y) 
		   =  \frac{2(x + 2k)}{\sqrt{2\pi}} f_{\delta'}(\sigma) \Big[ \frac{1}{P_1} - \frac{1}{P_2} \Big].
$$
and hence
\begin{equation}
\label{eq:I11lim}
I^{\infty,0}_{11} = \frac{\brc K}{B} 
	 \int_{\delta'}^{1-\delta'}  \sum_{k=-\infty}^{\infty}  \sum_{n=-\infty}^{\infty} R(k, n, \sigma, x, y) d\sigma.
\end{equation}
To complete Step 1, it remains to verify

\begin{lemma}
\label{lem:5.3}
Let $\mathfrak u(z) = \exp (-z) / z$. And let $P_1$ and $P_2$ be as defined above.  Then for $\delta \in \mathbb{R}$, $x \in [0, 1]$, $\sigma \in [0, 1]$, and $k, n$ not both $0$, the following sum converges uniformly in $\delta$, $x$, and $\sigma$ as $M \to \infty$.
$$
\sum_{k=-M}^{M} \sum_{n=-M}^{M}  (2k + x) \delta [\mathfrak u(P_1 \delta) - \mathfrak u(P_2 \delta)].
$$
\end{lemma}
\begin{proof}

Let us write 
$$
P_3 = (2k + x)^2 + (y-\sigma + 2n)^2, \quad P_4 = (2k + x)^2 + (y + \sigma - 2n)^2.
$$
We will show
\begin{equation}
\label{lemmasum1}
\lim_{M \rightarrow \infty} 
 \Bigg\{ \sum_{k: |k| > M } \sum_{n = 1}^{\infty} + \sum_{k =-M }^M \sum_{n = M}^{\infty} \Bigg\}
  \mathcal |S|= 0,
\end{equation}
where
$$
  \mathcal S = \mathcal S(k,n,\delta, \sigma, x,y) = (2k + x) \delta [\mathfrak u(P_1 \delta) -  \mathfrak u(P_2 \delta) 
+\mathfrak  u(P_3\delta) - \mathfrak  u(P_4 \delta)],
$$
and the convergence is uniform in $\delta, x, \sigma$. 
First, observe that
$$
P_1 - P_2 = -4 (\sigma + 2n) y, \quad  P_3 - P_4 = 4 ( 2n -  \sigma) y
$$
By the mean value theorem,
$$
\mathfrak u(P_1 \delta) - \mathfrak  u(P_2 \delta) = 
\mathfrak u'(P_1' \delta)(P_1 - P_2) \delta, \quad 
\mathfrak u(P_3 \delta) -  \mathfrak u(P_4 \delta) = 
\mathfrak u'(P_3' \delta)(P_3 - P_4)\delta
$$
for some $P_1' \in (P_1, P_2)$ and $P_3' \in (P_4, P_3)$.
Using the mean value theorem again, we conclude
\begin{align*}
&\mathfrak u(P_1 \delta) -  \mathfrak u(P_2 \delta) + 
\mathfrak u(P_3\delta) -  \mathfrak u(P_4 \delta) 
= - 4 \sigma y \delta [\mathfrak u'(P_1' \delta) + 
\mathfrak u'(P_3' \delta)] - 8ny \delta^2 (P_1' - P_3') 
\mathfrak u''(P_1'' \delta)
\end{align*}
for some $P_1'' \in (P_4, P_2)$.
In the sequel, $C$ denotes a universal constant (independent 
of $k,n,x,y,\delta, \sigma, L$ or any other parameters), whose
value is unimportant and can even change from line to line.
Now using the estimates $|\mathfrak u'(z)| < C/z^2$,  
$|\mathfrak u''(z)| < C/z^3$ for any real number $z$, we have
$$
|\mathcal S| \leq C \left( \frac{|k|}{(k^2 + n^2)^2} + \frac{|k|n^2}{(k^2 + n^2)^3} \right)
$$
Thus we conclude
$$
\sum_{k = M}^\infty \sum_{n= 1}^k |\mathcal S| \leq C \sum_{k = M}^\infty \sum_{n= 1}^k \frac{1}{k^3} \leq C/M
$$
and likewise
$$
\sum_{n = M}^\infty \sum_{k= 0}^{n-1} |\mathcal S| \leq C \sum_{n = M}^\infty \sum_{k= 0}^{n-1}  \frac{1}{n^3} \leq C/M
$$
We have verified \eqref{lemmasum1}. The lemma follows.
\end{proof}

{\bf Step 2: $I^{\infty,0}_{11}$ satisfies \eqref{eq:laplace3}}

We give two independent proofs for Step 2. The first proof is shorter and 
easily generalizes to the case of finite $t$. The second proof shows that the formulas
derived above are tractable (at least in case $t = \infty$). 

{\it Proof 1}: Step 1 shows that for any stochastic process $\cZ_t$ satisfying (H1) - (H3),
the limit \eqref{eq:I11lim} is the same. Recalling Remark \ref{rem}, 
we already have examples when (H1)-(H3) as well as the conclusion
of the theorem holds. Thus $I^{\infty,0}_{11}$ 
has to satisfy \eqref{eq:laplace3}. To finish the first proof, we identify the
constant $\varsigma$.

Let us consider the simplest possible random walk, called
the simple symmetric random walk. That is,
$w_1 = (0,-\sqrt 2)^T$, $w_2 = (0,\sqrt 2)^T$, $w_3 = (-\sqrt 2,0)^T$, $w_4 = (\sqrt 2,0)^T$
and
$$\cP(w_i) =  
 \frac14 \text{ for } i=1,...,4.
$$
In this case,
$\cL = (\sqrt 2 \integers)^2$ and by the central limit theorem, 
$\cZ_t / \sqrt t$ converges
to a $2$ dimensional standard normal random variable (we chose the normalization
$\sqrt 2$ so that the limiting covariance matrix is identity and so $\cZ$ fits into
the framework of Proposition \ref{prop:RW}).
In this case,
we clearly have $K=1$, $B= \sqrt 2$,
$\bA_1 = 1/4$ (and $\bB = 1$). Thus 
$\brc = c_1/4$. 
Next we claim that now $c_1 = 4 /\sqrt{\pi}$. To prove the claim, first note that
\begin{equation}
\label{eq:constRW1}
\lim_{T \to \infty} \sqrt T
\prob (\tau_0^{\cX} > T | \cX_0 = 0) = \frac{2}{\sqrt{\pi}} 
\end{equation}
(this follows from e.g.,
\cite[Proposition 5.1.2]{LL}). 
The proof of (H2) is based on the fact that under the assumption 
$\tau_0^{\cX} > T$, $\cZ_{\lfloor tT \rfloor} / \sqrt T$, $0 \leq t \leq 1$ converges
to a stochastic process whose first coordinate is a Brownian meander and the second coordinate
is a Brownian motion. Furthermore, the local limit theorem
also holds
under the assumption
$\tau_0^{\cX} > T$ which gives (H2) (see the details in Section
\ref{sec6}). This local limit theorem combined with
\eqref{eq:constRW1} gives 
$c_1 = \frac{2}{\sqrt{\pi}} \mbox{covol} (\cL) = 4 /\sqrt{\pi}$ which proves
the claim.

Thus in case
$K=1$, $B=\sqrt 2$, $\brc = 1 / \sqrt{\pi}$, \eqref{eq:I11lim} satisfies \eqref{eq:laplace3} with 
$\varsigma = 1$.  
Since \eqref{eq:I11lim} is linear in $\brc K / B$, we conclude that
in case of general $K,B$ and $\brc$, \eqref{eq:I11lim}
satisfies \eqref{eq:laplace3} with 
\begin{equation}
\label{eq:varsigma}
\varsigma = \frac{ \sqrt{2 \pi }\brc K}{ B}.
\end{equation}

{\it Proof 2}:

{\bf Step 2'a: $I^{\infty,0}_{11}$ is harmonic}

An elementary computation shows that $R(k, n, \delta, \sigma, x, y)$,
as a function of $x,y \in (0,1)^2$ is harmonic for any $k$ and $n$.  
Since the derivatives of $R(k, n, \sigma, x, y)$ with respect to $x$ and $y$ converge uniformly
in a neighborhood of $x,y$, the Laplacian can be taken inside the sum
in \eqref{eq:I11lim}. It follows that 
$I^{\infty,0}_{11}$ is harmonic.


{\bf Step 2'b: $I^{\infty,0}_{11}$ satisfies the boundary conditions
of \eqref{eq:laplace3}}

Recall \eqref{eq:I11lim} from Step 1. Let us first consider the case when $|n| + |k| >0$.  In this case, there is uniform convergence in $x,y$ and $\sigma$ so we can write the limit
inside the sum and the integral:
$$
\frac{\brc K}{B} 
 \sum_{k,n \in \integers; |n| + |k| > 0}
\int_{\delta'}^{1-\delta'}  \lim_{(x,y) \to (0,y_0)} 
R(k, n, \sigma, x, y) d\sigma.
$$
We can directly compute this limit as
\begin{align*}
& \int_{\delta'}^{1-\delta'} 
\lim_{(x,y) \to (0,y_0)} R(k,n, \sigma, x, y) d\sigma
		= \int_{\delta'}^{1-\delta'} R(k,n, \sigma, 0, y_0) d\sigma \\
		&= \int_{\delta'}^{1-\delta'}  f_{\delta'}(\sigma) \left[ \frac{16ky_0(\sigma + 2n)}{[(2k)^2 + (y_0-\sigma-2n)^2 ] [ (2k)^2 + (y_0+\sigma+2n)^2 ]}  \right] d\sigma.
\end{align*}
We see that for each $n$, these terms are antisymmetric in $k$, so that summing over $k$ and $n$, with $|n|+|k|>0$, all of the terms cancel. 
Now we consider the case $n=k=0$. This term gives: 

\begin{align*}
\lim_{(x,y) \to (0,y_0)} I_{11}^{\infty,0}
=  \frac{\brc K}{B} \frac{8}{\sqrt{2\pi}} 
\lim_{(x,y) \to (0,y_0)}
\int_{\delta'}^{1-\delta'}  f_{\delta'}(\sigma) \left[ \frac{\sigma x y}{[x^2 + (y-\sigma)^2 ] [ x^2 + (y+\sigma)^2 ]}  \right] d\sigma.
\end{align*}
To compute the above integral assume first that $\delta' < y_0 < 1-\delta'$, and decompose it as 
$$
\int_{\delta'}^{1-\delta'} ... d\sigma= 
\int_{y_0 - A x}^{y_0 - A x} ... d\sigma+ \int_{y \in [\delta', 1-\delta'] 
\setminus [y_0 - A x, y_0 + A x] } ... d\sigma =: I_{111} + I_{112}$$
for some large constant $A$.

First, we compute $I_{111}$. For $y_0$ and $A$ fixed, and
for $x$ and $|y-y_0|$ small, $yf_{\delta'}(\sigma)/[x^2 + (y + \sigma)^2]$ is close to 
$f_{\delta'}(y_0)/(4y_0)$
uniformly in $\sigma$ as in $I_{111}$. Indeed, this follows from the continuity of 
$f_{\delta'}$. Thus we can write this term 
in front of the integral. 
Now it remains to compute
$$\int_{y_0-Ax}^{y_0+A x} x \sigma / [x^2 + (y_0-\sigma)^2] d\sigma.$$
Let us apply the substitution $\rho = (\sigma - y_0 ) / x$. Then the previous integral becomes
$$\int_{-A}^A x \rho/(1 + \rho^2) d \rho + \int_{-A}^A y_0 /(1 + \rho^2) d \rho.$$
The first integral here is zero as the integrand is an odd function. The second integral is 
$\pi y_0 (1 + o_A(1))$. We conclude
\begin{equation}
\label{eq:I111}
\lim_{(x,y) \to (0,y_0)} I_{111} = \frac{\pi}{4} f_{\delta'}(y_0) (1 + o_A(1)).
\end{equation}

Next, we claim
\begin{equation}
\label{eq:I112}
\lim_{(x,y) \to (0,y_0)} I_{112} =  o_A(1).
\end{equation}
To prove \eqref{eq:I112}, we compute
\begin{align*}
 & \int_{y_0+Ax}^{1 - \delta'}  f_{\delta'}(\sigma) \frac{\sigma x y_0}{[x^2 + (y_0-\sigma)^2 ] [ x^2 + (y_0+\sigma)^2 ]}   					d\sigma \\
			&\leq \| f\|_{\infty} \sum_{i=1}^{\infty} \int_{y_0+A x i}^{y_0+A x (i+1)}  \frac{\sigma x y_0}{[x^2 + (y_0-\sigma)^2 ] 							[ x^2 + (y_0+\sigma)^2 ]}  d\sigma \\
			&\leq \| f\|_{\infty} \sum_{i=1}^{\infty} \int_{y_0+A x i}^{y_0+A x (i+1)}   \frac{\sigma x y_0}{								[ x^2 + (Ax i)^2 ][ 2y_0 \sigma ] 		}  d \sigma\\
			&\leq \frac{\| f\|_{\infty}}{2} \sum_{i=1}^{\infty} \int_{y_0+A x i}^{y_0+A x(i+1)}   \frac{x}{x^2 												[1 + A^2 i^2 ]}  d\sigma \\
			&= \frac{\| f\|_{\infty}}{2}  \sum_{i=1}^{\infty}  \frac{A}{ 1 + A^2 i^2 } 
\leq \frac{\pi^2\| f\|_{\infty}}{12} \frac{1}{A} .
\end{align*}
This estimate, combined with a similar computation for the domain $[\delta', y_0 - A \delta']$, verifies
\eqref{eq:I112}. Next, if $y_0 < \delta'$ or $y_0> 1-\delta'$, then clearly $I_{111} = 0$ and 
$I_{112} = o_A(1)$. 
Now combining \eqref{eq:I111} and \eqref{eq:I112}, we obtain the boundary conditions
of \eqref{eq:laplace3} on the "West side" (that is when $x = 0$) with the constant
$$
\varsigma = \frac{ \sqrt{2 \pi }\brc K}{ B}.
$$
which coincides with \eqref{eq:varsigma}.

Checking the boundary conditions on the other three sides is easier. First, recall that
$$
R(n,k,\sigma, x,y) = \frac{2(x + 2k)}{\sqrt{2 \pi}}
\frac{ (y+\sigma+2n)^2 - (y-\sigma-2n)^2 }{[(2k+x)^2 + (y - \sigma-2n)^2 ][(2k+x)^2 + (y + \sigma+ 2n)^2 ]}.
$$
Thus for every $k=0,1,2,...$, we have $R(n,k,\sigma, 1,y) = -R(n,-k-1,\sigma, 1,y)$ and so
$\sum_{k \in \integers}R(n,k,\sigma, 1,y)  = 0$ for every $n$. It follows that
$\lim_{x \to 1} I_{11}^{\infty,0} = 0$.
Clearly, $R(n,k,\sigma, x,0)  = 0$ for every $n$ and $k$ and so $\lim_{y \to 0} I_{11}^{\infty,0} = 0$.
Finally, to prove $\lim_{y \to 1} I_{11}^{\infty,0} = 0$, let us write
$$
\lim_{y \to 1} I_{11}^{\infty,0} = \sum_k \frac{2(x + 2k)}{\sqrt{2 \pi}} \sum_n  \frac{1}{P_1(n)}
- \frac{1}{P_2(n)},
$$
where
$P_1(n) = (2k+x)^2 + (1-\sigma-2n)^2$ and $P_2(n) =  (2k+x)^2 + (1+\sigma+2n)^2$.
Now observe that $P_2(n) = P_1(n+1)$. Thus the sum over $n$ is telescopic and so by absolute convergence, $\lim_{y \to 1} I_{11}^{\infty,0} = 0$.
We have finished the proof of Step 2'b. 
\end{proof}

Now we finish the proof of \eqref{eq:thm11}.
First note that Proposition \ref{prop1} implies
\eqref{eq:I11}. Thus it remains to verify \eqref{eq:I12}. Consider the following Dirichlet problem:
\begin{equation}
\label{eq:defU}
\begin{cases}
&\Delta U = 0 \text{ in } (0,1) \times (-1,2), \\
 &U(0,y) = \varsigma(f(y) - f_{\delta'}(y)),
U(1,y) = U(x,-1) = U(x,2) = 0,
\end{cases}
\end{equation}
where $f$ and $f_{\delta'}$ are identically zero on $[-1,0] \cup [1,2]$.
Now the proof of Proposition \ref{prop1} applied on the domain $(0,1) \times (-1,2)$
with boundary condition given by $f - f_{\delta'}$ implies that for any $\delta',x,y$ fixed,
$$
\lim_{\delta \to 0} \lim _{L \to \infty} I_{12} \leq U(x,y).
$$
Indeed, on the one hand if the particles are only killed upon leaving $(0,L) \times (-L,2L)$, 
then we obtain an upper bound on the number of surviving particles in case
when particles are killed upon leaving $(0,L) \times (0,L)$. On the other hand, the proof of 
Proposition \ref{prop1} is applicable on the larger domain since the boundary condition
is identically zero in a neighborhood of the corners.

Now since the function $f -  f_{\delta'}$ is supported on the union of two intervals
with total length $4 \delta'$ and is bounded uniformly in $\delta'$, we have
$\lim_{\delta' \to 0} U(x,y) = 0$ for all $x,y$ fixed. Thus
\eqref{eq:I12} follows and the proof of \eqref{eq:thm11} is complete.

The proof of \eqref{eq:thm12} is similar, so we only explain the differences.
First, the decomposition \eqref{eq:dec}  now reads
$$
\int_{\delta L^2}^{ t L^2 } ... dt + \int_{0}^{\delta L^2} ... dt  = : I_1 + I_2.
$$
In particular, $I_3$ is missing and $I_2$ is negligible as before.
We decompose $I_1 = I_{11} + I_{12}$ as before.
Proceeding as in Step 1 of the proof of Proposition
\ref{prop1}, we obtain
$$
\lim_{\delta \to 0} \lim_{L \to \infty} I_{11}
 = \frac{\brc K}{B} 
	 \int_{\delta'}^{1-\delta'}  \sum_{k=-\infty}^{\infty}  \sum_{n=-\infty}^{\infty} R(t,
k, n, \sigma, x, y) d\sigma,
$$
where
$$
	  R(t,k,n, \sigma, x, y) 
		   =  \frac{2(x + 2k)}{\sqrt{2\pi}} f_{\delta'}(\sigma) 
\left[ \frac{1}{P_1} \exp \left(- \frac{P_1}{2t}\right)- \frac{1}{P_2} \exp \left(- \frac{P_2}{2t}\right) \right].
$$
The first proof of Step 2 in Proposition
\ref{prop1} is the same as before. 
We prefer not to give a second proof of Step 2 as in the time dependent case, the 
formulas in Step 2'a become substantially longer.
Finally, the proof of \eqref{eq:I12} is again analogous to the previous case with $U$
as in \eqref{eq:defU} replaced by the unique solution $V(t,x,y): \reals_{\geq 0} \times
(0,1) \times (-1,2) \to \reals$ of 
$$
\begin{cases}
&V_t = \frac12 [V_{xx} + V_{yy}], \\
& V(t,0,y) = \varsigma(f(y) - f_{\delta'}(y)),
V(t,1,y) = V(t,x,-1) = V(t,x,2) = 0, V(0,x,y) = 0.
\end{cases}
$$
We have finished the proof of Theorem \ref{thm1}.

\section{Background on Lorentz gas}
\label{sec:billb}


\subsection{Preliminaries}

Here, we review some results for the 
Lorentz gas that are necessary to the proof of
Theorem \ref{thm2}. 
We refer the reader to \cite{ChM} for an in depth discussion.
Let us use the notation of Section \ref{sec:bil}. 

The map $\cF_0$ is hyperbolic in the sense that there are stable and
unstable cone fields 
$\cC^{u/s}_x \subset T_x \cM_0$ so that 
$D_x \cF_0 (\cC_x^u) \subset \cC^u_{\cF_0(x)}$ and 
$D_x \cF_0^{-1} (\cC_s^u) \subset \cC^2_{\cF_0^{-1}(x)}$
and for all $v \in \cC^u_x$, $\|D_x \cF_0(v)\| \geq \Lambda \|v\|$
(and likewise 
for all $v \in \cC^s_x$, $\|D_x \cF^{-1}_0(v)\| \geq \Lambda \|v\|$).
Furthermore, stable and unstable manifolds exist through almost every 
point, but not through every point because of singularities due to grazing
collisions. In fact, the presence of these singularities makes the study
of billiards particularly peculiar. 

Let us use the coordinates $(r,\varphi)$ on $\cM_0$ where $r$ is the 
arc length parameter of $\partial \cD_0$ and $\varphi \in [- \pi/2, \pi/2]$ 
is the angle between the postcollisional velocity and the normal vector to 
$\cD$. 
A curve $W \subset \cM_0$ is called unstable
if for every $x \in W$, $T_xW$ is in the unstable cone. Furthermore, 
an unstable curve $W$ is called weakly homogeneous if 
it does not intersect any singularity and 
there exists 
$k = 0,k_0,k_0+1,...$ so that for all $x = (r,\varphi)$ $\varphi \in
[(k+1)^{-2},k^{-2}]$ if $|k|> k_0$ or $|\varphi| < k_0^{-2}$. In other words, weakly
homogeneous unstable curves are required to be disjoint from the real singularities of $\cF_0$ as well as secondary singularities $\varphi = \pm k^{-2}$ for $|k| \geq k_0$. 
A weakly homogeneous unstable curve is called homogeneous
if it satisfies certain extra regularity properties whose exact form are not
needed for us (see the distortion and curvature bounds in \cite[Section 4.3]{CD09}).

A pair $\ell =(W, \rho)$ is called a {\it standard pair} if $W$ is a homogeneous
unstable curve and $\rho$ is a probability measure on $W$ so that
$$
\left|
\log  \frac{d \rho}{d \mbox{Leb}} (x) - \log  \frac{d \rho}{d \mbox{Leb}} (y)
\right|
\leq C_0 \frac{|W(x,y)|}{|W|^{2/3}}
$$
where $C_0$ is universal constant and $|.|$ stands for arc length.
Here and in the sequel $\log$ stands for logarithm with base $e$. We will
also use the notation $\log_2$ for the logarithm with base $2$.
Given $\ell$, we denote by $\nu_{\ell}$ the probability measure generated 
by $\rho$ and $\mbox{length}(\ell) = \mbox{length}(W)$. 
Due to the singularities, an image of a homogeneous unstable curve
will be a collection of unstable curves. Furthermore, the regularity
of $\rho$ is chosen in a way that is preserved by $\cF_0$. Thus the image
of a standard pair under $\cF_0$ is the weighted average of standard pairs.
Thus it is convenient to introduce the notion of a {\it standard family}:
a weighted average of standard pairs. Specifically, let us say that
$\cG = \{ \{ \ell_a = (W_a, \rho_a) \}_{a \in \mathfrak A}, \lambda \}$ is a standard family if $\ell_a$ are standard pairs, $W_a$'s are disjoint
and $\lambda$ is a probability measure on the index set $\mathfrak A$.
The standard family $\cG$ induces a measure $\nu_{\cG}$ on $\cM_0$ by
$$
\nu_{\cG}(B) = \int_{\mathfrak A} \nu_{\ell} (B \cap W_a) d \lambda(a)
$$
for Borel sets $B \subset \cM_0$.
For a given homogeneous unstable curve $W$, and for
$x \in W$, we denote by $r(x)$ the distance from $x$ to the closest endpoint of $W$, measured along $W$. We denote by $r_n(x)$ 
the distance from $\cF_0^n(x)$ to the closest endpoint of $W'$, where
$W'$ is the maximal homogeneous curve in the image $\cF^n(W)$
containing $\cF_0^n (x)$. We define the $Z$ function of a standard family
by 
$$
Z_{\cG} = \sup_{\eps >0} \frac{\nu_{\cG} (r < \eps)}{\eps}.
$$
Note that we assumed that the curves in a standard family are disjoint and 
so the function $r$ is well defined. Now we are ready to state the
last missing technical piece of Theorem \ref{thm2}: 
$\cG$ is any standard family with a finite $Z$ function. Examples include 
any standard pair or the invariant measure $\nu_0$.

A fundamental property of Sinai billiards is that the expansion wins over fragmentation. That is, most of the weight carried by the image of a standard pair is concentrated on long curves. The precise statement,
called Growth lemma is the following (see 
\cite[Prop. 4.9, 4.10]{CD09}):

\begin{lemma}
For any standard pair $\ell = (W, \rho)$ and any $n \in \integers_+$,
\begin{equation}
\label{eq:Markov}
\nu_{\ell} (A \circ \cF_0^n) = \sum_{i} c_{n,i} \nu_{\ell_{n,i}} (A),
\end{equation}
where $c_{n,i} >0$, $\sum_{i} c_{n,i} = 1$ and $\ell_{n,i} = (W_{n,i},
\rho_{n,i})$
are standard pairs so that $\cup_{i} W_{n,i} = \cF_0^n(W)$
and $\rho_{n,i}$ is a constant times the push-forward of $\rho$ by $\cF_0^n$.
Furthermore, there are universal constants $\varkappa,C$ so that
for any $n > \varkappa \log \mbox{length} (\ell)$ and for any $\eps >0$
$$
\sum_{i: \mbox{length}(\ell_{n,i}) < \eps} c_{n,i} < C \eps.
$$ 
\end{lemma}
We will refer to \eqref{eq:Markov} as Markov decomposition. A simple consequence of the
Growth lemma is the following lemma, which is proven in e.g. 
\cite[Proposition 7.17]{ChM}

\begin{lemma}
\label{lem:gl2}
There are constants $c_1$, $c_2$ and $\theta <1$ depending only on $\cD_0$
so that for any standard family $\cG$ with finite $Z$ function and for any $n$,
$$
Z_{\cF_0^n (\cG)} \leq c_1 \theta^n Z_{\cG} + c_2.
$$
\end{lemma}

Let 
 $\kappa: \cM_0 \to \reals^2$ be the free flight vector and 
$\check \kappa: \cM_0 \to \reals^2$ be the 
discrete free flight vector, that is 
$\check \kappa(q,v) = \Pi_{\integers^2}(\mathcal F_0(q,v)) - \Pi_{\integers^2}(q,v)$.
Let us also write $\bar \kappa = \int | \kappa| d \nu_0 \in \mathbb R_+$.

Let 
\begin{equation}
\label{def:Zcheck}
\check{\cZ}_n(q,v) = \sum_{j=0}^{n-1} \check \kappa(\cF_0^j(q,v)).
\end{equation}
Similarly to the flow, we write $\check{\cZ}_n = (\check{\cX}_n, \check{\cY}_n)$. Put
$$\tau_{0}^{\check{\cX}} = \min\{ n>0:  \check{\cX}_n < 0 \}$$ 
and for $x \neq 0$, put
$$\tau_{x}^{\check{\cX}} = \min\{ n>0:  \check{\cX}_n = x \}$$
(and likewise with $\check{\cX}$ replaced by $\check{\cY}$).  

The next result is the extension of the central limit theorem to a functional variant
in both discrete and continuous times (see e.g. \cite{C06}).

\begin{theorem} [Invariance principle]
Fix a standard pair $\ell$ and consider the stochastic processes $\tilde {\cZ}_{t}$,
$\check{\cZ}_n$ induced by the initial condition $\ell$. Then 
\begin{itemize}
\item[(a)] $\tilde{\cZ}_{tT} / \sqrt T$,
$t \in [0,1]$ converges weakly as $T \to \infty$ to a planar Brownian motion with 
zero mean and covariance
matrix $\Sigma$ (introduced in Section \ref{sec:bil}) uniformly 
for $\ell$ satisfying $| \log \mbox{length} (\ell) |> T^{1/4}$.
\item[(b)]
With the notation $\check \Sigma  = \bar \kappa \Sigma$ we have 
$\check{\cZ}_{\lfloor tN \rfloor} / \sqrt N$,
$t \in [0,1]$ converges weakly as $N \to \infty$ to a planar Brownian motion with 
zero mean and covariance
matrix $\check \Sigma$
uniformly 
for $\ell$ satisfying $| \log \mbox{length} (\ell) |> N^{1/4}$.
\end{itemize}
\end{theorem}

Another extension of the central limit theorem is the so called mixing local limit
theorem, which we discuss next.

\subsection{Mixing local limit theorem}

Recall \eqref{def:Zcheck}. Let us also define
$$
F_n (q,v) = \sum_{j=0}^{n-1} | \kappa( \cF_0^j (q,v))|.
$$
Given $\bx \in \reals^2$, $y \in \reals$ and 
a standard pair $\ell$ let us denote by $\vartheta_n$ the push-forward of $\nu_{\ell}$
by the map
$$
(q,v) \mapsto \left(
\check \cZ_n (q,v) - \langle \bx \sqrt n \rangle,
F_n(q,v) - n \bar \kappa - y \sqrt n,
\cF_0^n (q,v)
\right).
$$
That is, $\vartheta_n = \vartheta_n(\ell, \bx, y)$ is a measure on $\integers^2 \times \reals
\times \cM_0$. Fix an open set $A \subset \Omega_0$ 
as in \eqref{def:A} and define $\cA 
\subset \integers^2 \times \reals
\times \cM_0$ so that $((k,l), -t, (q,v)) \in \cA$ if and only if
$\Pi_{\integers^2}(q,v) = (k,l)$,
$\Pi_{\integers^2}(\Phi^t(q,v)) = 0$, $\Phi^t(q,v) \in A$ and $| \kappa(q,v) | >t$.
That is, $\cA$ contains phase points $(q +(k,l),v) $ and corresponding flight times $t$
so that a flight of length $t$ from $(q +(k,l),v) $ is free and arrives in the set $A$. By the
finite horizon assumption, $\cA$ is bounded. 

Let $\fg_{\Sigma}$ denote the Gaussian density with zero mean and covariance matrix
$\Sigma$.
The version of the mixing local limit theorem (MLLT) that we consider here is the following

\begin{theorem}
\label{thm:MLLT}
There is a positive definite $3 \times 3$ matrix $\tilde \Sigma$ whose top left $2\times 2$
minor is $\check \Sigma$ and constants $C, C_1, C_2$ so that for any standard
pair $\ell$ with $|\log \mbox{length} (\ell) | < n^{1/4}$ the following hold
\begin{enumerate}
\item[(a)] for any $(\bx, y) \in \reals^3$ and for any $A$ as in \eqref{def:A},
$$
\lim_{n \to \infty}
n^{3/2} \vartheta_n(\cA) 
= \mu_0(A) \bar \kappa \fg_{\tilde \Sigma} (\bx, y)
$$
uniformly for $\bx, y $ in compact subsets of $\reals^3$.
\item[(b)] for any $(\bx, y) \in \reals^3$ and for any $A$ as in \eqref{def:A} and for any 
positive integer $n$,
$$
n^{3/2} \vartheta_n(\cA) < C_1 \fg_{C \tilde \Sigma} (\bx, y) + C_2 ^{-1/2}.
$$
\end{enumerate}
\end{theorem}

A variant of Theorem \ref{thm:MLLT} was proved in \cite[Lemma 2.8]{DN16}. Specifically,
\cite[Lemma 2.8]{DN16} covers the case when $\check \cZ$ is replaced by $\check \cX$ 
and
$A = \Omega_0$ in the definition of $\vartheta_n$ (we included
the more general case of $A$ to accommodate for 
(DLE) as in Theorem \ref{Cor1}).
Since the proof directly applies here as well 
(except for one minor adjustment), we only discuss this minor adjustment and don't repeat
the entire proof.

\begin{proof}
First, we need some definitions. For a bounded H\"older function $f: \cM_0 \to \reals^d$,
we define $S(f)$ as the smallest closed additive subgroup of $\reals^d$ that supports the
values of $f-r$ for some $r \in \reals^d$. 
Let us write $f \sim g$ if $f$ and $g$ are cohomologous. That is,
$f(x) = g(x) + h(x) - h(\cF_0 (x))$ for a measurable $h$ and for all $x \in \cM_0$. We say that 
$f$ is minimal if $M(f) = S(f)$, where
$$
M(f) = \cap_{g \sim f} S(g).
$$
The only minor adjustment that is needed in the proof of \cite[Lemma 2.8]{DN16} is 
that we need to show that 
$$
f := ( \check \kappa, |\kappa| - \bar \kappa): \cM_0 \to \reals^3
$$ 
is minimal. That is, $M(f) = \integers^2 \times \reals$. 
(Heuristically, there is a clear obstruction to the MLLT in its present form 
if $M(f)$ is a proper subgroup of
$\integers^2 \times \reals$. It turns out that, similarly to the case of IID random variables,
this is the only possible obstruction.)
This generalizes 
\cite[Lemma A.3]{DN16}, which shows that 
$$
\tilde f := (\check \kappa_1, |\kappa| - \bar \kappa)
$$
is minimal, that is 
\begin{equation}
\label{eq:tf}
M(\tilde f) = \integers \times \reals.
\end{equation}
To establish the minimality of $f$, it is enough to prove the following. If $M(f)$ is a proper 
subgroup of $\integers^2 \times \reals$, then there are real numbers 
$\alpha,r$ and two measurable functions $h: \cM_0 \to \reals$, $g: \cM_0 \to \integers$
so that
\begin{equation}
\label{eq:A3}
|\kappa (q,v)| = h(q,v) - h(\cF_0(q,v)) + r + \alpha g(q,v).
\end{equation}
Indeed, a contradiction follows from \eqref{eq:A3} as in \cite{DN16}.
To prove \eqref{eq:A3}, 
we first recall that by \cite[Theorem 5.1]{SV04}, $\check \kappa$ is minimal. 
Thus the projection of $M(f)$ to the first two coordinates needs to be $\integers^2$. 
In particular, there exist $e_1 = (0,0,\alpha)^T$, $e_2 = (1,0,\beta)^T$ $e_3 = 
(0,1,\gamma)^T$ in $M(f)$. If $M(f)$ is a proper subgroup of $\integers^2 \times \reals$, 
then there exists a minimal $\alpha >0$ with the property that $e_1 \in M(f)$. Now we claim 
that $e_1,e_2,e_3$ generate $M(f)$. Indeed, by the choice of $\alpha$,
$e_1$ generates $M(f) \cap \{(0,0,z), z \in \reals \}$ and so $e_1, e_2, e_3$ generate
$$M(f) \cap \{ (x,y,z): (x,y) \in \{(1,0), (0,1) \} \}.$$ 
Since the projection of $M(f)$ to the first two coordinates in $\integers^2$, the claim follows.

Thus there are constants $r_1,r_2,r_3$ so that for every $(q,v) \in \cM_0$ there are integers
$m,n,k$ (depending on $(q,v)$) so that
\begin{align}
&(\check \kappa_1(q,v), \check \kappa_2(q,v), |\kappa (q,v)|)^T - 
(r_1, r_2, r_3)^T \nonumber \\
= &m e_1 + n e_2 + k e_3 + (h_1,h_2,h_3)^T(q,v) - 
(h_1,h_2,h_3)^T(\cF_0(q,v)) \label{eq:vec}
\end{align}
From the first coordinate of \eqref{eq:vec} we have
$$
n = \check \kappa_1 - r_1 - h_1 + h_1 \circ \cF_0
$$
and likewise from the second coordinate we we have
$$
k = \check \kappa_2 - r_2 - h_2 + h_2 \circ \cF_0
$$
Substituting these to the third coordinate of the equation \eqref{eq:vec}, we find
\begin{equation}
\label{eq:|kappa|}
|\kappa(q,v)| - \tilde r - \beta \check \kappa_1(q,v) - \gamma \check \kappa_2(q,v) = 
m \alpha 
+\tilde h(q,v)
- \tilde h ( \cF_0(q,v)),
\end{equation}
where $\tilde r = r_3 - r_1 \beta - r_2 \gamma$ and $\tilde h = h_3 - \beta h_1 - \gamma h_2$.
Fix now $(q,v)$ and write $\cF_0(q,v) = (q_1,v_1)$. Note that by reverting the free flight,
we have $\cF_0(q_1, - v_1) = (q, -v)$. Applying \eqref{eq:|kappa|} to $(q_1, -v_1)$, we obtain
\begin{equation}
\label{eq:|kappa|'}
|\kappa(q,v)| - \tilde r + \beta \check \kappa_1(q,v) + \gamma \check \kappa_2(q,v) = 
m' \alpha 
+\tilde h(q_1,-v_1)
- \tilde h ( \cF_0(q_1,-v_1)).
\end{equation}
Finally, adding \eqref{eq:|kappa|} to \eqref{eq:|kappa|'}, we obtain \eqref{eq:A3} with
$r = \tilde r$, $h(q,v) = \frac{1}{2} [\tilde h (q,v) + \tilde h (q_1 - v_1)]$ and
$g(q,v) = \frac{m + m'}{2}$. This completes the proof of \eqref{eq:A3}.
\end{proof}


\section{Proof of Theorem \ref{thm2}}
\label{sec6}

\subsection{Change of coordinates}
\label{sec:change}

Since $\cL$ is rational, we have 
$\mathfrak M := \Sigma^{1/2} \fl^{(K_1)} \in \integers^2$ and 
$\mathfrak N := \Sigma^{1/2} \fl^{(K_2)} \in \integers^2$. 
Furthermore,
$\mathfrak M$ and $\mathfrak N$ are primitive lattice vectors (i.e.
their coordinates are coprime due to the definition of $(K_1), (K_2)$).
Now we introduce an enlarged fundamental domain for the Lorentz gas. Let $Z'$
be the subset of $\integers^2$ containing the origin and those points of
$\integers^2$ that are in the interior of the parallelogram with vertices $0, \fM,
\fN,\fN+\fM$. Let 
$T' = \cup_{z \in Z'} [z-1/2, z+1/2]^2  / \sim$, where
$P \sim Q$ if $P-Q$ is in the lattice generated by $\fM, \fN$. That is, $T'$ is a
union of unit squares and $\sim$ is a pairing of all parallel sides on the boundary
of $T'$. In particular, $T'$ is a flat torus.
Now we put $\mathcal D_0' = T' \setminus \cup_{z \in Z'} 
\cup_{i=1}^{\mathfrak k} (B_i + z)$. See Figure \ref{fig:T'} for the special case
$\fM = (1,3)$ and $\fN = (2,1)$. $T'$ is the polygon with bold boundary
(modulo the identification). 

\begin{figure}
\begin{center}
\begin{tikzpicture}

\foreach \i in {0,...,3} 
\foreach \j in {0,...,4} {

            \draw[dotted] (\i-1/2,\j-1/2) -- (\i-1/2,\j+1/2) -- (\i+1/2,\j+1/2) -- (\i+1/2,\j-1/2) --(\i-1/2,\j-1/2);
\node at (\i, \j) [circle,fill,inner sep=0.5pt]{};

}

\draw[dashed] (0,0) -- (1,3) -- (3,4) -- (2,1) -- (0,0);

\draw[thick] (-1/2,-1/2) -- (1/2,-1/2) -- (1/2,1/2) -- (-1/2,1/2) -- (-1/2,-1/2);
\draw[thick] (1/2,1/2) -- (1/2,5/2) -- (3/2,5/2) -- (3/2,7/2) -- (5/2,7/2)
-- (5/2,3/2) -- (3/2,3/2) -- (3/2,1/2) -- (1/2,1/2);

\end{tikzpicture}
\caption{Enlarged fundamental domain} \label{fig:T'}
\end{center}
\end{figure}
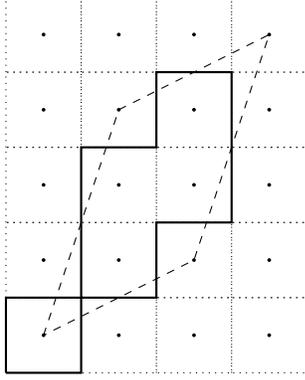

We are going to study the Sinai billiard in $\cD_0'$ and so we define
$\Phi_0'^t$, $\Omega_0'$, $\mu_0'$, $\cM_0'$, $\cF_0'$, $\nu_0'$ 
exactly as before using the larger configuration
space $\cD_0'$. 
Note that 
$\Phi_0^t$ is a factor of $\Phi_0'^t$ by the map $\iota: \Omega_0' \to \Omega_0$,
$\iota: (q,v) \mapsto (\bar q,v)$, where 
$ q \in \cD_0'$, $\bar q \in \cD_0$ and 
$\bar q = q (mod \integers^2)$. Also note that
$\Phi^t$ is an extension of both $\Phi_0^t$ and $\Phi_0'^t$. 

Given $(q,v) \in \Omega$, we write $\Pi'_{\integers^2}(q,v) = (m,n)$ if $q \in (m \fM, n \fN) + T'$
and $\Pi'_{\cD'_0}(q,v) = q_0$ if $q = q_0 + \Pi'_{\integers^2} (q,v)*(\fM, \fN)$,
where $*$ means multiplication coordinate-wise. 
Let us write
$\cZ'_t(q,v) = \Pi'_{\integers^2}(\Phi^t(q,v))$.

Note that for any $(k,l) \in \integers^2$ we can find a unique
$(k_0,l_0) \in Z'$ with $(k,l) \sim (k_0, l_0)$ and a unique $(m,n)$ so that
$(k,l) = (m \fM, n \fN) + (k_0, l_0)$. Let us write $[(k,l)] = (k_0,l_0)$
and $[[ (k,l)]] = (m,n)$. Note that 
\begin{equation}
\label{eq:largelattice}
[[\tilde{\cZ}_t(q,v)]] = \Sigma^{-1/2} (\tilde{\cZ}_t(q,v) - [\tilde{\cZ}_t(q,v)])
= {\cZ}'_t(q,v).
\end{equation}
Given $(q,v) \in \Omega$, we write $\Pi_{Z'}(q,v) = [\Pi_{\integers^2}(q,v)]$.
Let $\mathcal E_t(q,v)
= \Pi_{T'}(\Phi^t(q,v))=[\tilde{\cZ}_t(q,v)]$
($\mathcal E$ stands for extension). 
We will also write 
$[[ \ell' ]] = [[ \Pi_{\integers^2} (q,v) ]]$ 
for any $(q,v)$ in the support of $\nu_{\ell'}$ (we assume that the standard pairs are 
supported in one cell) and likewise $[[\cG']]$ for standard families.
All definitions and results in Section
\ref{sec:billb} extend to $\Phi_0'$. We will use those notations and results with 
a prime in the superscript.

\subsection{Proof of (H2)}
\label{sec:6.2}

We claim that (H2) follows from
\begin{enumerate}
\item[(H2')]

{\it
For any proper standard family $\cG'$ there is some $C_{\cG'}$ so that
for any $0 < \alpha < \beta$ and for any 
$0 < \eta, \gamma < \xi$ and for any $z' \in Z'$, if $[[\cG']] = 
(0, \lfloor \eta \sqrt T \rfloor )$, then
$$
\lim_{T \to \infty}  T^{3/2}
\nu_{\cG'} \left( 
\cZ'_T = \langle (\alpha, \gamma) \sqrt T \rangle, \mathcal E_T = z',
\min \{
 \tau_0^{\cY'},
\tau_{\xi \sqrt T }^{\cY'}, 
\tau_0^{\cX'},
\tau_{\beta \sqrt T }^{\cX'} \} > T
 \right)
$$
$$
= C_{\cG'} \psi(\alpha, \beta) \phi(\eta, \gamma, \xi)
$$
Furthermore, for any $\eps >0$, the convergence is uniform 
for $\eps < \alpha < \alpha + \eps < \beta < 1/\eps$,
$\eps < \eta < \eta + \eps < \xi < 1/\eps$,
$\eps < \gamma < \gamma + \eps < \xi$.
}
\end{enumerate}

To prove the claim, first recall that by \eqref{eq:largelattice},
$\cZ_t = \cZ_t' + \Sigma^{-1/2} \cE_t$.
To compare the initial conditions in (H2) and (H2'), 
note that given any standard family $\cG$ on $\cM_0$, there are exactly
$\bZ  := |Z'|$ corresponding standard families $\cG'_1,...,\cG'_{\bZ}$ on $\cM_0'$
that project to 
$\cG$ along $\iota$. Indeed, for any point $(q,v) \in \Omega_0$,
$\iota^{-1}((q,v)) = \{ (q + z',v), z' \in Z'\}$. Recall that the free flight is bounded by
$1$ and so the initial condition in (H2), i.e. $\cZ_0 = \fl$ and $\prob$
being induced by a standard family $\cG$, corresponds to an initial condition given
by $\cG'_{z'}$ for some $z' = 1,..., \bZ$ in (H2'). Indeed, the type of $\fl$ uniquely
defines $z'$. 
Thus $\cG$ and the type of $\fl$ in (H2')
is replaced by $\cG'$ in (H2). Since $\cE_t$ is bounded, the claim follows.

Note that for a given standard family $\cG$ and two lift ups 
$\cG'_{z'_1}$, $\cG'_{z'_2}$, $z'_1 \neq z'_2 \in Z'$,
the constants $C_{\cG'_{z'_1}}$, $C_{\cG'_{z'_2}}$ can be different. 
As we will see later, 
\begin{equation}
\label{eq:defCl}
C_{\cG'_{z'}} = \lim_{T \to \infty} \nu_{\cG'_{z'}}(\tau_0^{\cX'} > T)/ \sqrt T.
\end{equation}
Thus e.g. in Figure \ref{fig:T'}, 
$C_{\cG'_{(1,1)}} \geq C_{\cG'_{(1,2)}}$ for all standard families $\cG$. This inequality
is strict in case of some standard pairs. 
To prove this, note that in case of Figure \ref{fig:T'},
$\tau_0^{\cX'}(q,v) >T$ is equivalent to
$(\tilde{\cZ}_t)_2 \leq 3 (\tilde{\cZ}_t)_1$ for all $t \leq T$. Now observe that
$\tau_0^{\cX'}(q_0+(1,2),v) >T$ implies $
\tau_0^{\cX'}(q_0+(1,1),v) >T$, but the converse implication does not hold.

We will prove (H2'). The proof is build upon the results of \cite{DSV08,DN16}. 
In particular, 
\cite[Proposition 3.8]{DN16} gives that under the assumptions of (H2'),
\begin{equation}
\label{Prop3.8}
\lim_{T \to \infty}  T
\nu_{\cG'} \left( 
\cX'_T = \lfloor \alpha \sqrt T \rfloor, \mathcal E_T = z',
\min \{
\tau_0^{\cX'},
\tau_{\beta \sqrt T }^{\cX'} \} > T
 \right)
= C_{\cG'} \psi(\alpha, \beta) 
\end{equation}
with $C_{\cG'}$ defined by \eqref{eq:defCl}. 
Furthermore, \cite[Proposition 3.9]{DN16} gives that
under the assumptions of (H2'),
\begin{equation}
\label{Prop3.9}
\lim_{T \to \infty}  \sqrt T
\nu_{\cG'} \left( 
\cY'_T = \lfloor \gamma \sqrt T \rfloor, \mathcal E_T = z',
\min \{
\tau_0^{\cY'},
\tau_{\xi \sqrt T }^{\cY'} \} > T
 \right)
=  \phi(\eta, \gamma, \xi) 
\end{equation}
We interpret \eqref{Prop3.8} as the one dimensional version of (H2').
If the events on the left hand sides of \eqref{Prop3.8} and \eqref{Prop3.9}
were independent, then (H2') would follow immediately. By the invariance principle, 
$\cX_T'$ and $\cY_T'$ are asymptotically independent 
(since by the change of coordinates, the covariance
matrix is identity) but this yet is not enough to 
conclude (H2') as the events considered here have small probabilities. Thus we cannot
derive (H2') directly from \eqref{Prop3.8} and \eqref{Prop3.9}; we instead have to revisit their 
proofs. Since we only need to make minor changes to their proofs, we  
give details only at places where changes are needed and otherwise refer to 
\cite{DN16} (and sometimes give a sketch).


First we need some lemmas. Recall the notations introduced for the billiard ball
map in Section \ref{sec:billb}.
To simplify some notations a little, we will write
$$
\tau^{|\check \cX'|}_a =  \min \{ 
\tau^{\check \cX'}_a, \tau^{\check \cX'}_{-a} \}.
$$
and likewise for $\check \cX'$ replaced by $\check \cY'$.

\begin{lemma}
\label{lemma:highrec}
There are constant $C_3, C_4$ depending only on $\cD$ so that for every standard pair $\ell'$ with $[[ \ell' ]] = (0,0)$, for every $m > C_3 \log \mbox{length} (\ell)$ and for every $L$,
\begin{equation}
\label{eq:lem6.2}
\nu_{\ell'} \left(  \tau_{Lm}^{|\check{\cY'}|} <
 \tau_{m}^{|\check{\cX'}|}
\right) < 0.51^L + \frac{C_4 L}{m^{500}}.
\end{equation}
\end{lemma}

\begin{proof}
Let us fix a positive constant $\eta$ so that the probability that a 
standard planar Brownian motion
$W_t$ leaves the the box $[-1,1]^2$ through the North or South side (and not through
the East or West side) is at most $0.505$ whenever the $y$-coordinate of
$W_0$, denoted by $ (W_0)_2$, satisfies $|(W_0)_2| < \eta$. 
We are going
to use the invariance principle and the above estimate inductively $L$ times to derive the lemma. Each time the
North or South side is reached, we apply a Markov decomposition and discard
too short curves (hence the second term on the right hand side of \eqref{eq:lem6.2}).
Key to this argument is the fact that the limiting Brownian motion has a diagonal
covariance matrix, which is guaranteed by the change of coordinates from Section
\ref{sec:change}. Now we give the details of the proof.

Choosing $C_3$ large and using Lemma \ref{lem:gl2}, 
we can guarantee that the standard family
$\cG := \cF_0^{\eta m} (\ell')$
has a bounded $Z$ function (e.g. $Z_{\cG} < 2 c_2$, where $c_2$ is defined in Lemma
 \ref{lem:gl2}. Such standard families are sometimes called proper). Recall that we assumed
that the free flight is bounded by $1$. Thus for any standard pair $\ell'' = (W'', \rho'')$ in
$\cG$, $\| [[ \ell'']] \| \leq \eta m$. If $\mbox{length} (\ell'') < m^{-500}$, then we
estimate $\nu_{\ell''}(\cC) \leq 1$, where 
$\cC = \{ \tau_{Lm}^{|\check{\cY'}|} <
 \tau_{m}^{|\check{\cX'}|} \}$. By the Growth lemma, the measure carried by such standard pairs in $\cG$
is bounded by $C_4 m^{-500}$. Let us now assume that
$\mbox{length} (\ell'') > m^{-500}$. Then by the choice of $\eta$ and by 
invariance principle (assuming as
we can that $m$ is large enough), 
$$
\nu_{\ell''} (\tau_{m}^{|\check{\cY'}|} <
 \tau_{m}^{|\check{\cX'}|} ) \leq 0.51.
$$
Now let $\ell''' = (W''', \rho''')$ be a standard pair in the standard family 
$\cG_1 := \cF_0^{\tau_{m}^{|\check{\cY'}|}}(\ell'')$. 

Note that there
exists a constant $T_{\ell'''}$ so that 
for any $x\in W''$ with $\cF_0^{\tau_{m}^{|\check{\cY'}|}} \in W'''$,
$\tau_{m}^{|\check{\cY'}|} = T_{\ell'''}$. Indeed, this follows from the definition of 
homogeneous unstable curves. Now we distinguish two cases. Let us
say that $\ell'''$ is of type 1 if
$T_{\ell'''} > m$ or
$\mbox{length}(\ell''') < m^{-750}$. For type 1 standard pairs 
$\ell'''$, we use
the trivial bound $\nu_{\ell'''}(\cC) \leq 1$. By \cite[Lemma 5.1]{DN16},
the measure carried by standard pairs $\ell'''$ with $T_{\ell'''} >m^3$ 
is bounded by $C m^{-999}$. Thus by the growth lemma, the 
measure carried by standard pairs $\ell'''$ with $T_{\ell'''} \leq m^3$
and $\mbox{length}(\ell''') < m^{-750}$
is bounded by $C m ^{-747}$. Thus the total contribution of type 1
standard pairs is bounded by $C_4 m^{-500}$. Let us say that 
$\ell'''$ is of type 2 if not of type 1. By the invariance principle and 
by the definition of $\eta$, for every type 2 standard pair $\ell'''$, we have
$$
\nu_{\ell'''} (\tau_{2m}^{|\check{\cY'}|} <
 \tau_{m}^{|\check{\cX'}|} ) \leq 0.51.
$$
Thus we have derived
$$
\nu_{\ell'} (\tau_{2m}^{|\check{\cY'}|} <
 \tau_{m}^{|\check{\cX'}|} ) \leq 0.51^2 + \frac{2 C_4}{m^{500}}.
$$
Following the above procedure inductively, we obtain the lemma.

\end{proof}

\begin{lemma}
\label{lem:11.1}
For every $\delta >0$ and
for every $\xi >0$ there exists $M_0$ and $\brL$ so that for every standard pair $\ell'$ with $[[ \ell' ]] = (0,0)$
and $\mbox{length} (\ell') > \delta$, and for every $M>M_0$,
$$
\nu_{\ell'} \left( \tau_{\brL M}^{|\check{\cY'}|} <
\tau_{M}^{\check{\cX'}} \Big| \tau_{M}^{\check{\cX'}}  < \tau_{0}^{\check{\cX'}} 
\right) < \xi
$$
\end{lemma}

\begin{proof}
\cite[Lemma 11.1(a)]{DSV08} says that
\begin{equation}
\label{eq:brc}
\brc = \brc(\ell') = \lim_{M \to \infty} 
M \nu_{\ell'} (\tau_{M}^{\check{\cX'}} <\tau_{0}^{\check{\cX'}} )
\end{equation}
is finite. We will use the proof of that
lemma to prove our lemma. Let us recall the main steps of the proof.

Let $\bt_k = \tau_{2^k}^{\check{\cX'}}$ and 
$$
s_k = \min \{ n > \bt_k: \check{\cX'}_n < 0 \text{ or } \check{\cX'}_n  = 2^{k+1}\}.
$$
Let now $\ell''$ be a standard pair with 
\begin{equation}
\label{eq:stpaircond}
[[\ell'']]_1 = 2^k \text{ and } \mbox{length}(\ell'') > 2^{-100k}
\end{equation}
(we will consider $\ell''$ in the image of $\ell'$ under the map 
$(\mathcal F')^{\bt_k}$). The proof of \cite[Lemma 11.1(a)]{DSV08} is based on the following
identity (see \cite[Lemma 11.2]{DSV08}):
\begin{equation}
\label{eq:lemma11.2id}
\nu_{\ell''} \left(
\bt_{k+1} < \tau_{0}^{\check{\cX'}} \text{ and } 
r'_{\bt_{k+1}} \geq 2^{-100(k+1)}
\right) = \frac12 + O(2^{-k\zeta})
\end{equation}
with a universal positive constant $\zeta$.
Fixing an arbitrary $\eps >0$, one can choose $k_0$ large enough
so that 
an induction on $k=k_0,...,\log_2 M$ 
using \eqref{eq:lemma11.2id}
gives that 
\begin{equation}
\label{eq:approx}
|M \nu_{\ell'} (s_k = \bt_{k+1}, r'_{s_k} \geq 2^{-100(k+1)} \text{ for }k=k_0,...,\log_2 M) 
- \brc | < \eps,
\end{equation} 
which implies 
\eqref{eq:brc} (by the Growth lemma, the measure
of the points where $r_{s_k} < 2^{-100(k+1)}$ for some $k < \log_2 M$ can be neglected). We refer the reader to \cite{DSV08} for more details.

Now we turn to the proof of our lemma. Let us put 
$m_k = 2^k$, $\tilde k = (\log_2 M ) - k$ and
$$
L_k = 
\begin{cases}
2^k &\text{ if } k_0 \leq k < \frac12 \log_2 M \\
K 1.5^{\tilde k} &\text{ if } \frac12 \log_2 M \leq k < \log_2 M 
\end{cases}
$$
with some $K = K(\xi)$ to be specified later. Assuming that $k_0$ is
bigger than a universal constant (as we can), we have 
$m_k > 100 C_1 \log (1/m_k)$. Thus Lemma 
\ref{lemma:highrec} imply that for all standard pairs satisfying 
\eqref{eq:stpaircond}:
$$
\nu_{\ell''} \left( \min \{ \tau_{[[\ell'']]_2-L_km_k}^{\check{\cY'}},  \tau_{[[\ell'']]_2 + L_km_k}^{\check{\cY'}}\} <
\min \{ \tau_{0}^{\check{\cX'}},  \tau_{2m_k}^{\check{\cX'}}\}
\right) < 0.51^{L_k} + \frac{C_4 L_k}{m_k^{100}},
$$
which combined with \eqref{eq:lemma11.2id} gives
\begin{equation}
\label{eq:lemma11.2ext}
\nu_{\ell''} \left(
\bt_{k+1} < \min \{ \tau_{0}^{\check{\cX'}},
\tau_{[[\ell'']]_2 -L_k 2^k}^{\check{\cY'}},  \tau_{[[\ell'']]_2 + L_k 2^k}^{\check{\cY'}}
\}
\text{ and } 
r'_{\bt_{k+1}} \geq 2^{-100(k+1)}
\right) = \frac12 
+ E_{k, \ell''},
\end{equation}
where 
\begin{equation}
\label{eq:Ek}
-C' 2^{-k \zeta}-0.51^{L_k} - \frac{C_4 L_k}{m_k^{1000}} 
< E_{k, \ell''} \leq C' 2^{-k \zeta},
\end{equation}
with a universal constant $C'$.
Now we revisit the inductive proof of \eqref{eq:approx}. 
Let us write
\begin{equation}
\label{eq:cP}
\cP = 
\nu_{\ell'} \left( s_k = \bt_{k+1}, r'_{s_k} \geq 2^{-100(k+1)},
\tau^{|\check \cY'|}_{M + \sum_{j=k_0}^k L_j 2^j} > s_k
\text{ for }k=k_0,...,\log_2 M -1\right).
\end{equation}
Using \eqref{eq:lemma11.2ext} inductively, we find
$$
\cP = \nu_{\ell'} (\tau^{\check \cX'}_{2^{k_0}} < 
\min \{ \tau^{\check \cX'}_0,
\tau^{|\check{\cY'}|}_M \}) 
\prod_{k=k_0}^{\log_2 M -1}  \frac12 (1 + E_k),
$$
where $E_k$ satisfies the same inequalities \eqref{eq:Ek} as $E_{k, \ell''}$. As
before, choosing $k_0$ and $M$ large, we can guarantee
\begin{equation}
\label{eq:1}
 \cP > \frac{\brc - \xi'/10}{M}\prod_{k=k_0}^{\log_2 M -1}  (1 + E_k),
\end{equation}
where $\xi' = \xi \brc /2$.
Let us write
\begin{equation}
\label{eq:2}
 \prod_{k=k_0}^{\log_2 M -1}  (1 + |E_k|)
= \exp
\left(
\sum_{k=k_0}^{\log_2 M -1} \log (1 + |E_k|)
\right) \leq 
\exp
\left(
\sum_{k=k_0}^{\log_2 M -1}  |E_k|
\right) .
\end{equation}
Later we will show that 
\begin{equation}
\label{eq:xiin}
\sum_{k=k_0}^{\log_2 M} \left( 
C' 2^{- k \zeta} +
0.51^{L_k} + 
\frac{C_4 L_k}{m_k^{500}} \right) < \frac{\xi'}{10 \brc} = \frac{\xi}{20}.
\end{equation}

Before proving \eqref{eq:xiin}, let us show how it implies the lemma.
Combining \eqref{eq:1}, \eqref{eq:2} and \eqref{eq:xiin}, we find
\begin{equation}
\label{Plowerbd}
 \cP>  \frac{\brc - \xi'}{M}.
\end{equation}
Next observe that the event in \eqref{eq:cP}
implies that 
$$
\tau_{\tL M}^{|\check{\cY'}|}>
\tau_{M}^{\check{\cX'}},$$
where $\tL = 1+ \frac{1}{M} \sum_{k=k_0}^{\log_2 M} L_k 2^k$.
The next computation shows that 
$\tL$ is bounded by a constant $\brL = \brL(\xi)$ uniformly in $M$:
\begin{align*}
1 + \frac{1}{M} \sum_{k=k_0}^{\log_2 M} L_k 2^k &= 
1+ \frac{1}{M} \sum_{k=k_0}^{\frac12 \log_2 M} 4^k +
\frac{K}{M} \sum_{k=\frac12 \log_2 M}^{\log_2 M} 1.5^{\tilde k} 2^k \\
& \leq
5 + 
\frac{K}{M} \sum_{\tk =0}^{\frac12 \log_2 M} 1.5^{\tk} 2^{\log_2 M - \tilde k}
\leq 5 + K \sum_{\tk =0}^{\infty} \left( \frac34 \right)^{\tk} = 5 + 4K =:\brL.
\end{align*}
Thus we find
$$
\nu_{\ell'} \left( \tau_{\brL M}^{|\check{\cY'}|} <
\tau_{M}^{\check{\cX'}} \Big| \tau_{M}^{\check{\cX'}}  < \tau_{0}^{\check{\cX'}} 
\right) < 1 - \frac{ \cP}{\nu_{\ell'}(\tau_{M}^{\check{\cX'}}  < \tau_{0}^{\check{\cX'}} )} \leq 1 - \frac{\brc - \xi'}{\brc + \xi'} 
\leq \xi, $$
where the 
penultimate inequality uses \eqref{Plowerbd} and the last one uses the definition of $\xi'$. This proves the lemma. It remains to 
verify \eqref{eq:xiin}.

To prove \eqref{eq:xiin}, first choose $K = K(\xi)$ large, so that 
$$
\sum_{k=\frac12 \log_2 M}^{\log_2 M} 0.51^{L_k} < \sum_{\tilde k=0}^{\infty} 0.51^{K 1.5^{\tilde k}} < \frac{\xi}{100}.
$$
Then we compute
$$
\sum_{k=k_0}^{\log_2 M} C' 2^{-k \zeta} < \frac{\xi}{100},
$$
$$
\sum_{k=k_0}^{\frac12 \log_2 M} 0.51^{L_k} < \sum_{ k=k_0}^{\infty} 0.51^{2^k} < \frac{\xi}{100}
$$
and
$$
\sum_{k=k_0}^{\frac12 \log_2 M} \frac{L_k}{m_k^{500}} < 
\sum_{ k=k_0}^{\infty} 2^{-499k} < \frac{\xi}{100}.
$$
(Note that we can ensure the last inequality in all of the three displayed
formulas above by
increasing $k_0= k_0(\xi)$ if necessary.) Finally, we have
$$
\sum_{k=\frac12 \log_2 M}^{\log_2 M} \frac{L_k}{m_k^{500}} < 
\log_2 M \frac{1.5^{ \log_2 M}}{2^{250 \log_2 M}} = o(M^{-249}) < \frac{\xi}{100},
$$
which completes the proof of \eqref{eq:xiin}.
\end{proof}

\begin{lemma}
\label{lem:lastH2}
For every $\eta_1, \eta_2>0$ there exists $\eps_0$ so that for every $\eps< \eps_0$ and for every 
$\delta >0$ there is some $N_0$ so that for all $N > N_0$ and for all 
standard pair $\ell'$, with $[[ \ell' ]] = (0,0)$, $\mbox{length} (\ell) > \delta$, we have
$$
\nu_{\ell'} \left(
\tau_{\eps \sqrt N}^{\check{\cX'}} <  \min \{ \tau_{\eta_1 \sqrt N}^{\check{\cY'}} , 
\tau_{- \eta_1 \sqrt N}^{\check{\cY'}}, \eps N \}
\Big| \
 \tau_{0}^{\check{\cX'}} > N
\right) > 1 - \eta_2.
$$
\end{lemma}

\begin{proof}
\cite[Lemma 5.2]{DN16} implies that 
$$
\nu_{\ell'} \left(
\tau_{\eps \sqrt N}^{\check{\cX'}} < \eps N,
\Big| 
 \tau_{0}^{\check{\cX'}} > N
\right) > 1 - \frac{\eta_2}{2}.
$$
and \cite[Theorem 8]{DSV08} implies that 
\begin{equation}
\label{eq:thm8}
 \lim_{T \to \infty} \nu_{\ell'}(\tau_0^{\check{\cX}'} > N)/ \sqrt N =:\check C_{\ell'}
\end{equation}
is finite for all standard pairs and non-zero for some.
Thus it suffices to prove
\begin{equation}
\label{eq:thm82ndcoord}
\nu_{\ell'} (\cA \cB \cC) < \frac{\eta_2 \check C_{\ell'}}{4 \sqrt N},
\end{equation}
where
$$
\cA = \{ 
\tau_{\eps \sqrt N}^{\check{\cX'}} >  \min \{ \tau_{\eta_1 \sqrt N}^{\check{\cY'}} , \tau_{- \eta_1 \sqrt N}^{\check{\cY'}} \}
\},
\quad \cB = \{ \tau_{\eps \sqrt N}^{\check{\cX'}} < \eps N \}, \quad 
\cC = \{
 \tau_{0}^{\check{\cX'}} > N \}.
$$
To prove \eqref{eq:thm82ndcoord}, let us write 
$$
\cD = \{ 
\tau_{\eps \sqrt N}^{\check{\cX'}} < \tau_{0}^{\check{\cX'}}  \}
$$
and 
$$
\nu_{\ell'} (\cA \cB \cC)
= \nu_{\ell'} (\cA \cB \cC \cD)  \leq \nu_{\ell'} (\cA \cD) \nu_{\ell'} (\cC| \cA \cB \cD) =: I*II.
$$
To estimate $II$, we use Markov decomposition at time $\tau_{\eps \sqrt N}^{\check{\cX'}}$.
By the invariance principle, $II$ is asymptotic (as $N \to \infty$) to the probability that the 
maximum of the
standard Brownian motion before time $1$ is less than $\epsilon$ which is bounded from above
by $\hat c \eps$.
 Let 
$\brc = \brc(\ell') $ as in \eqref{eq:brc}
and let $\xi = \frac{\eta_2\check C_{\ell'}}{4 \brc \hat c}$. Lemma \ref{lem:11.1} gives
$\brL = \brL(\xi)$. Then we choose $\eps_0 < \eta_1 /\brL$. Now Lemma \ref{lem:11.1} implies that
$$I = \nu_{\ell'} (\cA |\cD) \nu_{\ell'} (\cD) \leq \xi \frac{\brc}{\eps \sqrt N}$$
and so \eqref{eq:thm82ndcoord} follows.

\end{proof}

Next, we have
the following extension of \cite[Theorem 3.5]{DN16} to two dimensions.

\begin{proposition}
\label{prop:3.5}
The process $\check{\cZ}'_{tN}/(\sqrt{ \bar \kappa N})$, $0 < t < 1$ induced by the measure $\nu_{\cG'}(.| \tau_0^{\cX'} > N)$
converges weakly as $N \to \infty$ to the planar stochastic process with independent coordinates, whose
first coordinate is a Brownian meander and the second coordinate is a standard Brownian motion.
\end{proposition}

The proof of Proposition \ref{prop:3.5} is the same as that of 
\cite[Theorem 3.5]{DN16} except that \cite[Lemma 5.2]{DN16}
is replaced by our Lemma \ref{lem:lastH2}.
The sketch of the proof is as follows. 
Under the assumption $\tau_{0}^{\check{\cX'}} > N$,
with high probability, we have
$\tau_{\eps \sqrt N}^{\check{\cX'}} <  \min \{ \tau_{\eta_1 \sqrt N}^{\check{\cY'}} , 
\tau_{- \eta_1 \sqrt N}^{\check{\cY'}}, \eps N \}$. Then we use the 
invariance principle starting at time $\tau_{\eps \sqrt N}^{\check{\cX'}}$.
The invariance principle is applicable since 
$\nu_{\ell''}(\tau_{0}^{\check{\cX'}} > N)$ is bounded from below for
$\ell''$ with $\ell'' > \delta_0$ and $[[\ell'']]_1 = \eps \sqrt N$ for fixed $\eps$. Thus we
obtain a planar 
Brownian motion with identity covariance matrix, whose first coordinate starts from $\eps$ and
does not reach $0$ before time $1$ and whose second coordinate 
starts from a position with absolute value less than $\eta_1$. Choosing 
$\eta_1$ small (and consequently 
$\eps$ small), the distribution of this process is close to the one described
in the lemma. 

(H2') is a local version of Proposition \ref{prop:3.5} in continuous time.
The proof of (H2') is again analogous to the one dimensional case given in 
\cite[Proposition 3.8]{DN16}. Although the
proof is quite lengthy, let us a give a short sketch.
Let $N = T/ \bar \kappa$, $N_1 = (1 -\delta_t) N$ with a small $\delta_t$
and partition the rectangle
$R_T := [0, \beta \sqrt T] \times [0, \xi \sqrt T]$ into boxes $B_k$
with side length $\delta_s \sqrt T$ with some fixed
$\delta_s$ small. 
Proposition
\ref{prop:3.5} gives the asymptotic probability 
(for $T$ large, other parameters fixed)
of arriving in a box $B_k$ after discrete time 
$N_1$. Then 
for any given box $B_k$ and any given a standard pair $\ell'$ in this box
as initial condition
(with $\mbox{length}(\ell') > \delta_0$ for some fixed $\delta_0$), 
we need to find the probability that in the remaining
continuous time before $T$ but after the first $N_1$ collisions, the 
particle arrives in 
the cell $\langle \alpha \sqrt T, \gamma \sqrt T \rangle$.
To give an upper bound, we use the MLLT by simply ignoring the requirement that, 
in the remaining $\approx \delta_t T$ time, the particle has to stay inside $R_T$. Switching from discrete to continuous time is a non-trivial step. 
For "typical" number of collisions, Theorem \ref{thm:MLLT}(a) is used,
whereas the contribution of non-typical number of collisions is negligible
by Theorem \ref{thm:MLLT}(b). This gives the upper bound in (H2').
To prove the lower bound, one needs to verify that the error made 
by ignoring the requirement that the particle has to stay inside $R_T$
for the last $\approx \delta_t T$ time is negligible. If a particle leaves
$R_T$ and returns to $\langle \alpha \sqrt T, \gamma \sqrt T \rangle$,
then in particular it has to travel a distance 
$\min \{ \alpha, 1-\alpha, \gamma, 1-\gamma\} \sqrt T$
during time $\delta_t T$. This has small probability which gives the lower bound
in (H2') (in \cite{DN16} $\delta_t$ is chosen small given $\alpha \in (0,1)$,
now we need to choose it small given $\alpha, \gamma \in (0,1)$). No other
substantial change is required.

\subsection{Proof of (H3)}

As in case of (H2), we use the change of coordinates to reformulate 
(H3) as

{\it (H3')
For any $x \in (0,1)$ and $y \in (-1,1)$, and for any 
proper standard family $\cG'$ with $[[\cG']] = (0,0)$
$$
\lim_{\delta \to 0} \lim_{L \to \infty} \int_{[0, \delta L^2] \cup [L^2/\delta, \infty)}
L \nu_{\cG'}(
\cZ'_t = \langle (xL,yL)  \rangle,
\min \{ 
\tau_0^{\cX'},
\tau_{L }^{\cX'} \} > t
)dt = 0.
$$
}
The fact that (H3') implies (H3) follows the same way as we proved that
(H2') implies (H2). In fact, this case is easier as contrary
to the case of (H2). We only need an upper bound here and so we can ignore
the requirement $\cE_t = z'$ at the cost of losing a constant multiplier.

As in the upper bound of (H2'), we can derive that for given 
$(x,y) \in (0,1)^2$ and $\eps >0$ there exists $\delta$
so that for large enough $L$ and for any $t < \delta L^2$,
$$
\nu_{\cG'}(
\cZ'_t = \langle (xL,yL)  \rangle,
| \tau_{x L /2}^{\cX'} < \tau_0^{\cX'}) < \frac{\eps}{L^2}.
$$
Using this estimate, the proof of (H3') follows as in \cite[Lemma 7.2]{DN16}.

\section*{Acknowledgement}
This work started when both authors were at University of Maryland, College Park.
PN was partially supported by NSF DMS 1800811 and NSF DMS 1952876.

\end{document}